\documentclass[onecolumn,nonote,noversion]{cdcarticle}

\def\MODE{3}

\usepackage{hyperref}
\hypersetup{hidelinks}

\usepackage[utf8]{inputenc}						% handle accents properly
\usepackage[english]{babel}						% handle hyphenation properly

\usepackage{math}
\usepackage{enumerate}
\usepackage{tikz}
\usepackage{lipsum}
\usepackage{cite}

\graphicspath{{./graphics/}}

\newcommand{\rt}[1]{\textcolor{red}{#1}}
\newcommand{\LL}[1]{#1}
\DeclareMathOperator{\IQC}{IQC}

\if\MODE1
\makeatletter
\def\@IEEEsectpunct{.\ \,}
\def\paragraph{\@startsection{paragraph}{4}{\z@}{1.5ex plus 1.5ex minus 0.5ex}%
{0ex}{\bfseries}}
\makeatother
\fi

\begin{document}

\title{Exponential Stability Analysis via\\ Integral Quadratic Constraints}

\if\MODE1
\author{Ross~Boczar,~\IEEEmembership{Member,~IEEE,}
        Laurent~Lessard,~\IEEEmembership{Member,~IEEE,}
        Andrew~Packard,~\IEEEmembership{Member,~IEEE,}
        and~Benjamin~Recht,~\IEEEmembership{Member,~IEEE}% <-this % stops a space
\thanks{Ross Boczar and Benjamin Recht are with with the Department
of Electrical Engineering and Computer Science, University of California, Berkeley, CA~94720, USA. Email: \texttt{\{boczar,brecht\}@berkeley.edu}.}
\thanks{Andrew Packard is with the Department
of Mechanical Engineering, University of California, Berkeley, CA~94720, USA. Email: \texttt{apackard@berkeley.edu}.}% <-this % stops a space
\thanks{Laurent Lessard is with the Department of Electrical and Computer Engineering, University of Wisconsin--Madison, Madison, WI~53706, USA. Email: \texttt{laurent.lessard@wisc.edu}}% <-this % stops a space
\thanks{Manuscript received ...\rt{TODO}}}

\markboth{IEEE Transactions on Automatic Control TODO}%
{\rt{TODO}}
\else
\author{Ross~Boczar, \hfill
        Laurent~Lessard, \hfill 
        Andrew~Packard, \hfill
        Benjamin~Recht}
\note{}
\fi
\maketitle

%%%%%%%%%%%%%%%%%%%%%%%%%%%%%%%%%%%%%%%%%%%%%%%%%%%%%%%%%%%%%%%%%%%%%%%%%%%%%%%%
\begin{abstract}
The theory of integral quadratic constraints (IQCs) allows verification of stability and gain-bound properties of systems containing nonlinear or uncertain elements. Gain bounds often imply exponential stability, but it can be challenging to compute useful numerical bounds on the exponential decay rate.
This work presents a \LL{generalization} of the classical IQC results of Megretski and Rantzer~\cite{megretski_system_1997} that leads to a tractable computational procedure for finding exponential rate certificates that are far less conservative than ones computed from $L_2$ gain bounds alone.
An expanded library of IQCs for certifying exponential stability is also provided and the effectiveness of the technique is demonstrated via numerical examples.
\end{abstract}

\if\MODE1\fi

%%%%%%%%%%%%%%%%%%%%%%%%%%%%%%%%%%%%%%%%%%%%%%%%%%%%%%%%%%%%%%%%%%%%%%%%%%%%%%%%
%!TEX root = boczar_submitted.tex

\section{Introduction}\label{sec:introduction}

\if\MODE1\IEEEPARstart{A}{nalysis}\else Analysis \fi in the context of robust control is generally concerned with obtaining absolute performance guarantees about a system in the presence of bounded uncertainty. Examples of such results include the small gain theorem \& passivity theory~\cite{zames}, dissipativity theory~\cite{willems721}, the structured singular value $\mu$~\cite{ref:doyle82}, and integral quadratic constraints (IQCs)~\cite{megretski_system_1997}.

In this paper, we present a modification of IQC theory, the most general of the aforementioned tools, that allows one to certify \emph{exponential stability} rather than just bounded-input bounded-output (BIBO) stability. Moreover, we can compute numerical bounds on the exponential decay rate of the state.

Even when BIBO stable systems are exponentially stable, estimates of the exponential decay rates provided by standard IQC theory are typically very conservative. We will show that this conservatism can be greatly reduced if we directly certify exponential stability and use the method presented herein to compute the associated decay rate.

Our modified IQC analysis was successfully applied in~\cite{lessard_analysis_2014} to analyze convergence properties of commonly-used optimization algorithms such as the gradient descent method. These algorithms converge at an exponential rate when applied to strongly convex functions, and the modified IQC analysis automatically produces very tight bounds on the convergence rates.
Another potential application is in time-critical systems. In embedded model predictive control, for example, it is vital to have robust guarantees that desired error bounds will be met in the allotted time without overflow errors and in spite of fixed-point arithmetic. See \cite{mpc} and references therein.

\paragraph*{A special case} While a general treatment of exponential bounds is provided in the sequel, it is worth noting that exponential stability can be proven directly for some special cases. 
%
%As previously noted, exponential stability certificates are often conservative when they are derived from $L_2$ gain bounds. However, it is well known that exponential stability can be proven directly in some special cases.
To illustrate this fact, consider a linear time-invariant (LTI) discrete-time plant $G$ with state-space realization $(A,B,C,D)$. Suppose $G$ is connected in feedback with a \textit{strictly-input passive} nonlinearity $\Delta$. A sufficient condition for BIBO stability is that there exists a positive definite matrix $P \succ 0$ and a scalar $\lambda \ge 0$ satisfying the linear matrix inequality (LMI)
\begin{equation}\label{eq:passiv}
\addtolength{\arraycolsep}{-0pt}
\bmat{A & B \\ I & 0}^\tp\bmat{P & 0 \\ 0 & -P}\bmat{A & B \\ I & 0}\\
+ \lambda \bmat{0 & C^\tp \\ C & D\!+\!D^\tp} \prec 0
\end{equation}
This result is also related to the Positive Real Lemma (see~\cite{passivity_dissipativity_relationship} and references therein).
If we define $V(x) \defeq x^\tp P x$, then~\eqref{eq:passiv} implies that $V$ decreases along trajectories:
$
V(x_{k+1}) \le V(x_k)
$
for all $k$. BIBO stability then follows from positivity and boundedness of $V$.
Observe that when~\eqref{eq:passiv} holds, we may replace the right-hand side by $-\epsilon P$ for some sufficiently small $\epsilon > 0$. We then conclude that
$
V(x_{k+1}) \le (1-\epsilon)V(x_k)
$
for all $k$ and exponential stability follows. We may then maximize $\epsilon$ subject to feasibility of~\eqref{eq:passiv} to further improve the rate bound.

Unfortunately, the approach outlined above of including $-\epsilon P$ fails in the general IQC setting due to the different role played by $P$ in the associated LMI. In IQC theory, the LMI comes from the Kalman-Yakubovich-Popov (KYP) lemma and although it is structurally similar to~\eqref{eq:passiv}, $P$ is not positive definite in general and $V$ may not decrease along trajectories.

Our key insight is that by suitably modifying both the LMI \emph{and} the IQC definition, we obtain a more broadly applicable condition for certifying exponential stability.

The paper is organized as follows. We cover some related work in the remainder of the introduction, we explain our notation and some basic results in Section~\ref{sec:prelim}, we develop and present our main result in Section~\ref{sec:main}, and we discuss computational considerations in~Section~\ref{sec:computation}. An explicit construction of the (conservative) rate guarantees implied by finite $L_2$ gain is given in~Section~\ref{sec:L2exp}. In~Section~\ref{sec:library} we provide a library of applicable IQCs. Finally, we present illustrative examples demonstrating the usefulness of our result in Section~\ref{sec:examples}, and we make some concluding remarks in Section~\ref{sec:conclusion}.

\paragraph*{Related work}
It is noted in~\cite{megretski_system_1997,rantzer_system_1994} that BIBO stability often implies exponential stability. In particular, exponential stability follows if the nonlinearity satisfies an additional \emph{fading memory} property. \LL{So under mild assumptions, the robust stability guarantee from IQC theory automatically implies exponential stability as well. The proof of this result uses the $L_2$ gain from the stability analysis to construct an exponential rate bound. We will see in Section~\ref{sec:examples} that bounds computed in this way can be very conservative.}

Other proofs of exponential stability have appeared in the literature for specific classes of nonlinearities. Some examples include sector-bounded nonlinearities~\cite{corless_bounded_1993,konishi_robust_1999} and nonlinearities satisfying a Popov IQC~\cite{jonsson_nonlinear_1997}. These works exploit LMI modifications akin to the one shown with~\eqref{eq:passiv} earlier in this section.

This work is inspired by~\cite{lessard_analysis_2014}, which presents an approach for proving the robust exponential stability of optimization algorithms. The approach of~\cite{lessard_analysis_2014} uses a time-domain formulation of IQCs modified to handle exponential stability. \LL{In contrast, the present work develops the aforementioned exponential stability analysis entirely in the frequency domain and its applicability is not restricted to the analysis of iterative optimization algorithms. Moreover, we clarify the connection to the seminal IQC results in~\cite{megretski_system_1997}.} Parts of this work first appeared in the conference paper~\cite{boczar2015exponential}. Since then, an analogous continuous-time formulation with alternative techniques and motivations also appeared in \cite{huseiler}.
%!TEX root = boczar_submitted.tex

\section{Notation and preliminaries}\label{sec:prelim}

We adopt a setup analogous to the one used in~\cite{megretski_system_1997}, with the exception that we will work in discrete time rather than continuous time.
%
%Discrete-time signals are indexed by subscripts: $x_k$ is the signal $x$ at time $k$.
%
%The conjugate of a complex number $z\in\C$ is denoted $\bar z$ and
%
The conjugate transpose of a vector $v\in\C^n$ is denoted $v^*$. The unit circle in the complex plane is denoted $\T \defeq \set{z\in\C}{|z|=1}$
% Note that $\T$ is parameterized by $z = e^{j\omega}$ with $\omega\in [0,2\pi)$.
The $z$-transform of a time-domain signal $x\defeq (x_0,x_1,\dots)$ is denoted $\hat x (z)$ and defined as
$
\hat x(z) \defeq \sum_{k=0}^\infty x_k z^{-k}
$. The $i$-th coordinate of the vector $x$ is denoted $x^{(i)}$.

A Hermitian positive definite (semidefinite) matrix $M$ is denoted $M \succ 0$ ($M \succeq 0$).
Function composition is denoted $(g\circ f)(x):=g(f(x))$.
A sequence $u = (u_0,u_1,\dots)$ is said to be in $\ltwo$ if $\sum_{k=0}^\infty|u_k|^2<\infty$. A sequence $u_k$ is said to be in $\ltwo^\rho$ for some $\rho \in (0,1)$ if the sequence $(\rho^{-k}u_k)$ is in $\ltwo$, i.e. $\sum_{k=0}^\infty\rho^{-2k}|u_k|^2<\infty$. Note that $\ltwo^\rho \subset \ltwo$.
Let $\RHinf^{m\times n}$ be the set of $m\times n$ matrices whose elements are proper rational functions with real coefficients analytic outside the closed unit disk.

Consider the standard setup of Fig.~\ref{fig:lure1} (the \emph{Lur'e system}). The block $G$ contains the known LTI part of the system while $\Delta$ contains the part that is uncertain, unknown, nonlinear, or otherwise troublesome.
  \begin{figure}[thpb]
    \centering
    \includegraphics{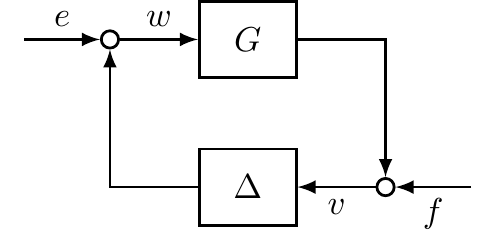}
    \caption{Linear time-invariant system $G$ in feedback with a nonlinearity $\Delta$.
    }\label{fig:lure1}
  \end{figure}
  The interconnection is said to be \emph{well-posed} if the map $(v,w)\mapsto(e,f)$ has a causal inverse. The interconnection is said to be bounded-input bounded-output (BIBO) stable if, in addition, there exists some $\gamma > 0$ such that when $G$ is initialized with zero state,
  \begin{equation*}\label{eq:gain}
  \norm{v}^2 + \norm{w}^2 \le \gamma \bl( \norm{e}^2 + \norm{f}^2 \br)
  \end{equation*}
  for all square-summable inputs $f$ and $e$, and where $\norm{\cdot}$ denotes the $\ltwo$ norm.
  Finally, the interconnection is \emph{exponentially stable} if there exists some $\rho \in (0,1)$ and $c > 0$ such that if $f=0$ and $e=0$, the state $x_k$ of $G$ will decay exponentially with rate $\rho$. That is,
  \[
  \norm{x_k} \le c\, \rho^k\, \norm{x_0}
  \qquad\text{for all }k.
  \]
We now present the classical IQC definition and stability result, which will be modified in the sequel to guarantee exponential convergence. These results are discrete-time analogs of the main IQC results of Megretski and Rantzer~\cite{megretski_system_1997}.

\begin{defn}[IQC]\label{def:iqc}
  Signals $y \in \ltwo$ and $u \in \ltwo$  with associated $z$-transforms  $\hat y(z)$ and $\hat u(z)$  satisfy the \emph{IQC} defined by a Hermitian complex-valued function $\Pi$ if
  \begin{equation}\label{iqceq}
    \int_\T\,
    \bmat{ \hat y(z)\\ \hat u(z) }^* \Pi(z)
    \bmat{ \hat y(z)\\ \hat u(z) } dz \geq 0\:.
  \end{equation}
A bounded causal operator $\Delta$ satisfies the IQC defined by $\Pi$ if~\eqref{iqceq} holds for all $y\in\ltwo$ with $u = \Delta(y)$. We also define $\IQC(\Pi(z))$ to be the set of all $\Delta$ that satisfy the IQC defined by $\Pi$.
\end{defn}
\begin{thm}[Stability result]
\label{thm:classic}
Let $G(z) \in \RHinf^{m\times n}$ and let $\Delta$ be a bounded causal operator. Suppose that:
  \begin{enumerate}[i)]
    \item for every $\tau \in [0,1]$, the interconnection of $G$ and $\tau \Delta$ is well-posed.
    \item for every $\tau \in [0,1]$, we have $\tau \Delta \in \IQC(\Pi(z))$.
    \item there exists $\epsilon > 0$ such that
      \begin{equation*}
          \begin{bmatrix}
          G(z)\\I
          \end{bmatrix}^*
          \Pi(z)
          \begin{bmatrix}
          G(z)\\I
          \end{bmatrix} \preceq -\epsilon I, \quad \forall z\in \T\:.
      \end{equation*}
  \end{enumerate}
Then, the feedback interconnection of $G$ and $\Delta$ is BIBO stable.
\end{thm}
%!TEX root = boczar_submitted.tex

\section{Frequency-domain condition}\label{sec:main}
In this section, we augment Definition~\ref{def:iqc} and the classical result of Theorem~\ref{thm:classic} to derive a frequency-domain condition that certifies exponential stability.

\begin{defn}
The operators $\rho_+,\:\rho_-$ are defined as the time-domain, time-dependent multipliers $\rho^k, \rho^{-k}$, respectively, where $\rho\in(0,1)$ is a defined constant.
\end{defn}
\begin{rem}\label{rem:rho}
The operator $\rho_- \circ (G(z) \circ \rho_+)$ is equivalent to the operator $G(\rho z)$. This follows from the fact that, for any constant $a>0$ and signal $u_k$, the $z$-transform of $a^{-k}u_k$ is given by $\hat u(az)$. See Fig.~\ref{fig:rho_comp} for an illustration.
\begin{figure}[thpb]
  \centering
  \includegraphics{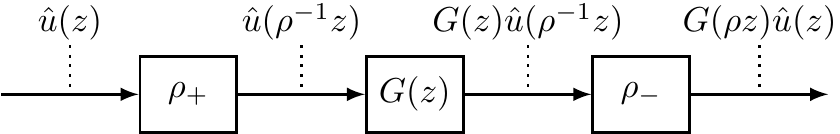}
  \caption{Illustration of Remark \ref{rem:rho}.}
  \label{fig:rho_comp}
\end{figure}
\end{rem}
In order to show exponential stability of the system in Fig.~\ref{fig:lure1}, we will relate it to BIBO stability of the modified system shown in Fig.~\ref{fig:lure2}. This equivalence is closely related to the theory of stability multipliers \cite{safonov_zames-falb_2000}.

\begin{figure}[thpb]
  \centering
  \includegraphics{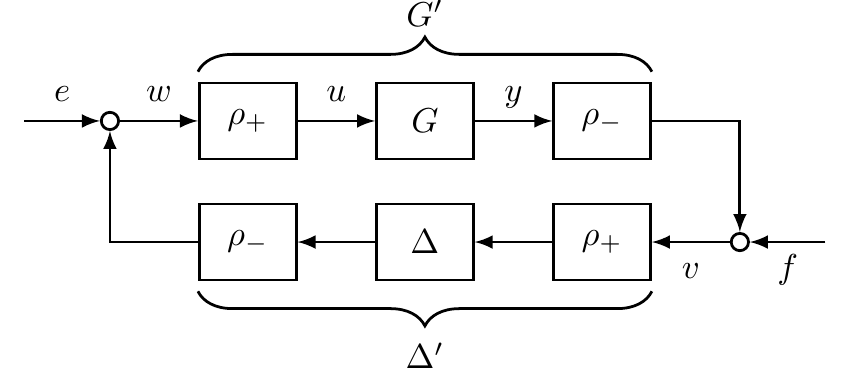}
  \caption{Modified feedback diagram with additional multipliers and inputs. For appropriately chosen $e$ and $f$ and with zero initial condition, we show how this diagram is equivalent to that of Fig.~\ref{fig:lure1}.}
  \label{fig:lure2}
\end{figure}

\begin{prop}\label{prop:exp}
  Suppose $G(z)$ has a minimal realization $(A,B,C,D)$. If the interconnection in Fig.~\ref{fig:lure2} is BIBO stable, then the interconnection in Fig.~\ref{fig:lure1} with initial state~$x_0$ is exponentially stable.
  
\end{prop}
\begin{proof}
Intuitively, if $v$ and $w$ are \emph{small} in the BIBO sense compared to $e$ and $f$, then $y$ must be even smaller. \if\MODE1
A complete proof is included in the Appendix.
\else
See Appendix~\ref{sec:prop5proof} for a detailed proof.
\fi
  \end{proof}

In an effort to define IQCs for the transformed system shown in Fig.~\ref{fig:lure2}, we introduce the concept of the \emph{$\rho$-IQC}.
\begin{defn}[$\rho$-IQC]\label{def:piqc}
  Signals $y \in \ltwo^\rho$ and $u \in \ltwo^\rho$  with associated $z$-transforms  $\hat y(z)$ and $\hat u(z)$  satisfy the \emph{$\rho$-IQC} defined by a Hermitian complex-valued function $\Pi$ if
  \begin{equation}\label{rr}
  \int_\T\,
  \begin{bmatrix} \hat y(\rho z)\\ \hat u(\rho z) \end{bmatrix}^* \Pi(\rho z)
  \begin{bmatrix} \hat y(\rho z)\\ \hat u(\rho z) \end{bmatrix} dz \geq 0\:.
  \end{equation}
  A bounded causal operator $\Delta$ satisfies the $\rho$-IQC defined by $\Pi$ if~\eqref{rr} holds for all $y\in\ltwo^\rho$ with $u = \Delta(y)$. We also define $\IQC(\Pi(z),\rho)$ to be the set of all $\Delta$ that satisfy the $\rho$-IQC defined by $\Pi$. 
\end{defn}

Note that the concept of a $\rho$-IQC generalizes that of a regular IQC. Indeed, we have $\IQC(\Pi(z),1) = \IQC(\Pi(z))$. The restriction of $u \in \ltwo^\rho$ and $y \in \ltwo^\rho$ corresponds to the restriction of $u \in \ltwo$ and $y \in \ltwo$ in the classical definition of IQC \cite{megretski_system_1997}. 
Now equipped with $\rho$-IQCs, we can relate $\Delta'$ in Fig.~\ref{fig:lure2} to  $\Delta$ in Fig.~\ref{fig:lure1}.

\begin{prop}\label{prop:iqcequiv}
Let $\Delta$ be a bounded causal operator, and let $\Pi$ be a Hermitian complex-valued function. As in Fig.~\ref{fig:lure2}, define $\Delta' \defeq \rho_- \circ (\Delta \circ \rho_+)$. Then the following statements are equivalent.
\begin{enumerate}[(i)]
  \item $\Delta \in \IQC(\Pi(z),\rho)$
  \item $\Delta' \in \IQC(\Pi(\rho z) )$
\end{enumerate}
\end{prop}

\begin{proof}
We define the discrete Fourier transform of the input and output of $\Delta$ as $\hat y(z)$ and $\hat u(z)$, respectively. Then, from the definition of $\rho_+$ and $\rho_-$, we have that $\hat w(z) = \hat u(\rho z)$ and $\hat v(z) = \hat y(\rho z)$. Substituting into the IQC definition~\eqref{iqceq}, we obtain~\eqref{rr} as required.
\end{proof}
Proposition~\ref{prop:iqcequiv} is illustrated in Fig.~\ref{fig:rho_equiv}.
\begin{figure}[thpb]
  \centering
  \includegraphics{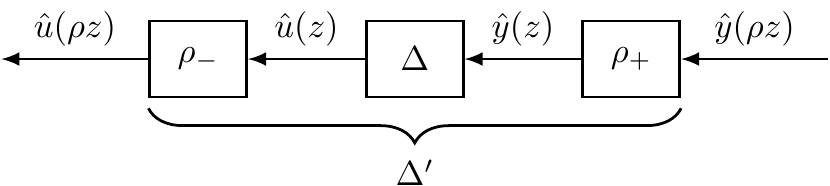}
  \caption{Illustration of Proposition \ref{prop:iqcequiv}.}
  \label{fig:rho_equiv}
\end{figure}

We now state our main result, an exponential stability theorem analogous to the classical result in Theorem~\ref{thm:classic}.

\begin{thm}[Exponential stability]
\label{thm:main}
Fix $\rho \in (0,1)$. Let $G(\rho z) \in \RHinf^{m\times n}$ and $\Delta$ be a bounded causal operator such that $\Delta':=\rho_- \circ (\Delta \circ \rho_+)$ is also bounded and causal. Furthermore, suppose that:
  \begin{enumerate}[i)]
    \item for every $\tau \in [0,1]$, the interconnection of $G$ and $\tau \Delta$ is well-posed.
    \item for every $\tau \in [0,1]$, we have $\tau\Delta \in \IQC(\Pi(z),\rho)$.
    \item there exists $\epsilon > 0$ such that
      \begin{equation}\label{eq:expfdi}
          \begin{bmatrix}
          G(\rho z)\\I
          \end{bmatrix}^*
          \Pi(\rho z)
          \begin{bmatrix}
          G(\rho z)\\I
          \end{bmatrix} \preceq -\epsilon I, \quad \forall z\in \T\:.
      \end{equation}
  \end{enumerate}
Then, the interconnection of $G$ and $\Delta$ shown in Fig.~\ref{fig:lure1} is exponentially stable with rate $\rho$.
\end{thm}

\begin{proof}
We apply Theorem \ref{thm:classic} to the interconnection in Fig.~\ref{fig:lure2} with operators $G(\rho z)$ and $\Delta'$ and the IQC $\Pi(\rho z)$.

\if\MODE1
\begin{enumerate} [(a)]
  \item Since Fig.~\ref{fig:lure1} and Fig.~\ref{fig:lure2} have the same interconnection structure, well-posedness is equivalent. 
  \item Due to the equivalence of IQCs in Proposition \ref{prop:iqcequiv},
  \begin{align*}
  &\hspace{-3em}\tau \Delta \in \IQC(\Pi(z),\rho)\\
  &\; \iff \rho_- \circ ((\tau \Delta) \circ \rho_+) \in \IQC(\Pi(\rho z)) \\
  &\; \iff \tau (\rho_- \circ (\Delta \circ \rho_+)) \in \IQC(\Pi(\rho z))  \\
  &\; \iff \tau \Delta' \in \IQC(\Pi(\rho z)) \:.
  \end{align*}
  \item This is condition iii) of Theorem \ref{thm:classic} using $G(\rho z)$ and $\Delta'$.
\end{enumerate}
\else
\begin{enumerate} [(a)]
  \item Since Fig.~\ref{fig:lure1} and Fig.~\ref{fig:lure2} have the same interconnection structure, well-posedness is equivalent. 
  \item Due to the equivalence of IQCs in Proposition \ref{prop:iqcequiv},
  \begin{align*}
  \tau \Delta \in \IQC(\Pi(z),\rho)
  &\iff \rho_- \circ ((\tau \Delta) \circ \rho_+) \in \IQC(\Pi(\rho z)) \\
  &\iff \tau (\rho_- \circ (\Delta \circ \rho_+)) \in \IQC(\Pi(\rho z))  \\
  &\iff \tau \Delta' \in \IQC(\Pi(\rho z)) \:.
  \end{align*}
  \item This is condition iii) of Theorem \ref{thm:classic} using $G(\rho z)$ and $\Delta'$.
\end{enumerate}
\fi
Thus, these three conditions ensure BIBO stability of the system in Fig.~\ref{fig:lure2}. We then apply Proposition~\ref{prop:exp} to arrive at exponential stability of Fig.~\ref{fig:lure1}.

Note that the assumption $G(\rho z) \in \RHinf^{m\times n}$ restricts us to verifying rates that are no faster than the rate of convergence of the open-loop $G$, which corresponds to the largest (in magnitude) pole of $G(z)$. Assuming WLOG that $\Delta(0)=0$, this is clear as $\Delta\equiv0$ (corresponding to open-loop $G$) satisfies any $\rho$-IQC.
\end{proof}
%!TEX root = boczar_submitted.tex

\section{Computation}\label{sec:computation}

As in the classical IQC setting, to guarantee stability, the frequency-domain inequality (FDI) \eqref{eq:expfdi} must be verified for every $\omega \in [0,2\pi)$. However, if the IQC in question exhibits a particular factorization, then the discrete-time KYP Lemma can be applied to convert the infinite-dimensional FDI to a finite-dimensional LMI. We now review these results.

\begin{defn}\label{defn:factor}
We say $\Pi$ has a \emph{factorization} $(\Psi,M)$ if
\begin{equation*}
  \Pi(z) = \Psi(z)^*M \Psi(z)\:,
\end{equation*}
where $\Psi$ is a stable linear time-invariant system, $M$ is a constant Hermitian matrix, and $\Psi(z)^*$ denotes the conjugate transpose of $\Psi(z)$.
\end{defn}
\LL{\begin{rem}Definition~\ref{defn:factor} is similar to J-spectral factorization (see \cite{jspectral} and references therein), except we require them to hold for arbitrary $z\in\C$. Spectral factorizations are commonly evaluated on the unit circle for discrete systems (c.f. the imaginary axis for continuous-time systems). In such cases, we have $z^* = z^{-1}$ for all $z\in\T$ and $s^* = -s$ for all $s\in j\R$. For this reason, factorizations are conventionally written using the \emph{para-Hermitian conjugate} defined as $\Psi^\sim(z) \defeq \Psi^\tp(z^{-1})$ (c.f. $\Psi^\sim(s) \defeq \Psi^\tp(-s)$ for continuous time). Although these definitions are equivalent to $\Psi(z)^*$ (c.f. $\Psi(s)^*$) in general, we cannot use the para-Hermitian conjugate for our factorization because we require it to hold for all $z\in\C$.
\end{rem}}

\begin{rem}\label{rem:rhoIQC_time_domain}
If $\Pi(z)$ has a factorization $(\Psi,M)$ and $\Psi(\rho z)$ is stable, then by Parseval's Theorem,~\eqref{rr} is equivalent to 
  \begin{equation*}
    \sum_{k=0}^\infty \rho^{-2k}z_k^\tp  M z_k \geq 0\:,
\quad\text{where }
    z\defeq \Psi \begin{pmatrix}y\\u\end{pmatrix} \:.
  \end{equation*}
\end{rem}

The KYP lemma, stated below, is attributed to Kalman, Yakubovich, and Popov. A simple proof and further references can be found in~\cite{rantzer_KYP}.

\begin{lem}[Discrete-time KYP Lemma]\label{lem:kyp}
  Suppose $A$, $B$, $M$ are given matrices where $M$ is Hermitian and $A$ has no eigenvalues on the unit circle. Then the following FDI:
  \begin{equation*}
  \begin{bmatrix}
  (zI-A)^{-1}B\\I
  \end{bmatrix}^* M 
  \begin{bmatrix}
  (zI-A)^{-1}B\\I
  \end{bmatrix} \prec 0
  \end{equation*}
  holds for all $z\in\T$ if and only if there exists a $P=P^\tp$ and $\lambda \ge 0$ satisfying the LMI
  \begin{equation*}
  \begin{bmatrix}
  A&B\\I&0
  \end{bmatrix}^\tp
  \begin{bmatrix}
  P&0\\0&-P
  \end{bmatrix}
  \begin{bmatrix}
  A&B\\I&0
  \end{bmatrix} + \lambda M \prec 0\:.
  \end{equation*}
\end{lem}

\begin{cor}
Suppose the realization of $G$ is given by $(A,B,C,D)$ and assume $\Pi$ has a factorization $(\Psi,M)$, where the realization of $\Psi$ is given by
\begin{equation*}
\Psi = \left[\begin{array}{c|cc}
   A_\Psi & B_{\Psi_1} & B_{\Psi_2} \\ \hlinet
  C_\Psi & D_{\Psi_1} & D_{\Psi_2}
  \end{array}\right].
\end{equation*}
Then \eqref{eq:expfdi} is equivalent to the existence of $P=P^\tp $ and $\lambda\ge 0$ such that
  \begin{equation}\label{eq:explmi}
      \begin{bmatrix}  
    \hat A^\tp P\hat A-\rho^2 P & \hat A^\tp P\hat B\\ \hat B^\tp P \hat A & \hat B^\tp  P \hat B
    \end{bmatrix}
     + 
    \lambda \begin{bmatrix}
    \hat C^\tp \\ \hat D^\tp 
    \end{bmatrix}
    M
    \begin{bmatrix}
    \hat C & \hat D
    \end{bmatrix}
     \prec 0
  \end{equation}
  where $(\hat A, \hat B, \hat C, \hat D)$ are defined as
  \begin{equation*}
  \stsp{\hat A}{\hat B}{\hat C}{\hat D} \defeq
  \left[\begin{array}{cc|c}
    A & 0 & B \\
    B_{\Psi_1} C & A_\Psi & B_{\Psi_2}+B_{\Psi_1}D \\ \hlinet
    D_{\Psi_1}C & C_\Psi & D_{\Psi_2} + D_{\Psi_1}D
  \end{array}\right]\:.
  \end{equation*}
  
\end{cor}
\begin{proof}
A similar result is proven in \cite{seiler}, which we repeat here for completeness.
\begin{equation*}
\begin{bmatrix}
      G(z)\\I
      \end{bmatrix}^*
      \Pi(z)
      \begin{bmatrix}
      G(z)\\I
  \end{bmatrix}
= \\
  \begin{bmatrix}\star \end{bmatrix}^*
      M
      \begin{bmatrix}
      \hat C & \hat D
      \end{bmatrix}
      \begin{bmatrix}
      (zI-\hat A)^{-1}\hat B\\I
      \end{bmatrix}
\end{equation*}
where $\star$ denotes the repeated part of the quadratic form surrounding $M$.
Similarly, we have
\if\MODE1
\begin{align*}
  &\rho^{-2}\begin{bmatrix}
      G(\rho z)\\I
      \end{bmatrix}^*
      \Pi(\rho z)
      \begin{bmatrix}
      G(\rho z)\\I
  \end{bmatrix}\\
  &\qquad=
  \begin{bmatrix}\star\end{bmatrix}^*
      \rho^{-2}M
      \begin{bmatrix}
      \hat C & \hat D
      \end{bmatrix}
      \begin{bmatrix}
      (\rho zI-\hat A)^{-1}\hat B\\I
      \end{bmatrix}\\
  &\qquad= 
   \begin{bmatrix}\star\end{bmatrix}^*
    \rho^{-2}M
    \begin{bmatrix}
    \hat C & \hat D
    \end{bmatrix}
    \begin{bmatrix}
    (zI-\rho^{-1}\hat A)^{-1} \rho^{-1}\hat B\\I
    \end{bmatrix}\:.
\end{align*}
\else
\begin{align*}
  \rho^{-2}\begin{bmatrix}
      G(\rho z)\\I
      \end{bmatrix}^*
      \Pi(\rho z)
      \begin{bmatrix}
      G(\rho z)\\I
  \end{bmatrix}
  &=
  \begin{bmatrix}\star\end{bmatrix}^*
      \rho^{-2}M
      \begin{bmatrix}
      \hat C & \hat D
      \end{bmatrix}
      \begin{bmatrix}
      (\rho zI-\hat A)^{-1}\hat B\\I
      \end{bmatrix}\\
  &= 
   \begin{bmatrix}\star\end{bmatrix}^*
    \rho^{-2}M
    \begin{bmatrix}
    \hat C & \hat D
    \end{bmatrix}
    \begin{bmatrix}
    (zI-\rho^{-1}\hat A)^{-1} \rho^{-1}\hat B\\I
    \end{bmatrix}\:.
\end{align*}
\fi
If $\rho^{-1}\hat A$ has no eigenvalues on the unit circle, we may then invoke Lemma~\ref{lem:kyp} (applied to $\rho^{-1} \hat A$, $\rho^{-1} \hat B$, and the appropriate $M$ term) and multiply through by $\rho^2$ to show that \eqref{eq:expfdi} is equivalent to the existence of $P=P^\tp$ and $\lambda \ge 0$ such that \eqref{eq:explmi} holds, as required.
\end{proof}
With the advent of fast interior-point methods to solve LMIs, the feasibility of the LMI \eqref{eq:explmi} can often be quickly ascertained for any fixed $\rho^2$. Since the size of the LMI is often on the order of the size of the system $G$ and the IQC $\Pi$, many practical linear systems lead to LMIs of relatively moderate size.

Finding the best upper bound amounts to minimizing~$\rho^2$ subject to~\eqref{eq:explmi} being feasible. This type of problem occurs frequently in robust control and is known as a \emph{generalized eigenvalue optimization problem} (GEVP)~\cite{boyd_linear_1997}. The GEVP is not an LMI because~\eqref{eq:explmi} is not jointly linear in $\rho^2$ and $P$. One simple approach to solving the GEVP is to perform a bisection search on $\rho^2$, but there are more sophisticated methods available; see for example~\cite{boyd_elghaoui}.

\begin{rem}
The results above may also be carried through in continuous time. In that case, an equation analogous to \eqref{eq:expfdi} must be satisfied for $G(s-\lambda)$ for all $\omega \in [0,\infty)$, and can be verified by finding $P=P^\tp$ and $\lambda\ge 0$ such that
\begin{equation*}
  \begin{bmatrix}  
  \hat A^\tp P+ P \hat A-2\lambda P & P\hat B\\ \hat B^\tp P & 0
  \end{bmatrix} + 
  \lambda
  \begin{bmatrix}
  \hat C^\tp \\ \hat D^\tp 
  \end{bmatrix}
  M
  \begin{bmatrix}
  \hat C & \hat D
  \end{bmatrix}
   \prec 0\:.
\end{equation*}
An alternative continuous-time formulation is detailed in \cite{huseiler}.
\end{rem}

Applying a bisection search on $\rho^2$ requires the $\rho$-IQC to obey a certain monotonicity property, which we now define.

\begin{defn}[Monotonicity]
We say an IQC $\Pi(z)$ satisfies the \emph{monotonicity property} if for all $0 < \rho \le \rho' < 1$, we have:
\begin{align*}
\Delta \in IQC(\Pi(z),\rho) \implies \Delta \in IQC(\Pi(z),\rho').
\end{align*}
\end{defn}
All of the $\rho$-IQCs discussed herein satisfy the monotonicity property. If an IQC does not satisfy this property, then a grid search may be used instead of bisection.

%!TEX root = boczar_submitted.tex

\section{Exponential rates from gain bounds}\label{sec:L2exp}

In \cite{megretski_system_1997}, IQC analysis is used to certify $L_2$ stability of interconnected systems. As noted in~\cite{megretski_system_1997}: ``for general classes of ordinary differential equations, exponential stability is equivalent to the input/output stability...''.

While input/output stability often implies exponential stability, we will show through examples that exponential rates constructed from $\ell_2$ bounds can be very conservative. This fact justifies the use of a dedicated technique for certifying exponential rates rather than using an $\ell_2$ analysis.

We will need two results. First, a well-known generalization of Theorem~\ref{thm:classic} that allows us to optimize the $\ltwo$ gains over any pair of signals. We'll consider the scenario of Fig.~\ref{fig:lure_aug}, which is slightly more general than the setup in Fig.~\ref{fig:lure1}.

\begin{figure}[th]
	\centering
	\includegraphics{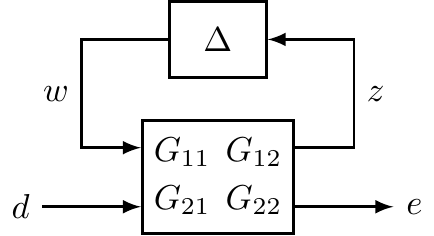}
	\caption{Augmented LTI system $G$ in feedback with a nonlinearity $\Delta$.}\label{fig:lure_aug}
\end{figure}
\noindent We would like to show that the input $d$ and output $e$ satisfy some IQC of the form
\begin{equation}\label{iqcperf}
\int_\T\,
\bmat{ \hat d(z)\\ \hat e(z) }^* \Pi_p(z)
\bmat{ \hat d(z)\\ \hat e(z) } dz \geq 0\:.
\end{equation}
The following result appears for example in~\cite{apkarian2006iqc} and a complete proof is given in~\cite{erin_thesis}.

\begin{thm}\label{thm:classic_performance}
Let $G(z) \in \RHinf^{m\times n}$ and let $\Delta$ be a bounded causal operator. Suppose $G$ is partitioned according to the dimensions of the input and output channels in Fig.~\ref{fig:lure_aug}. Suppose the interconnection of $G_{11}$ and $\Delta$ is well-posed and stable and $\Delta \in \IQC(\Pi(z))$. If there exists $\epsilon > 0$ such that  
\begin{equation*}
\begin{bmatrix}
\star
\end{bmatrix}^*
\begin{bmatrix}
\Pi(z) & 0 \\ 0 & -\Pi_p(z)
\end{bmatrix}
\begin{bmatrix}
G_{11}(z) & G_{12}(z) \\ I & 0 \\ 0 & I \\ G_{21}(z) & G_{22}(z)
\end{bmatrix} \preceq -\epsilon I \quad \forall z\in \T
\end{equation*}
then for all $d\in\ltwo$ and $e\in\ltwo$, Equation \eqref{iqcperf} is satisfied.
\end{thm}

\begin{rem}[see \cite{erin_thesis}]\label{rem:dont_need_stability}
In Theorem~\ref{thm:classic_performance}, if $\Pi_{p,22} \preceq 0$ and $(G_{11},\Delta)$ satisfies assumptions (i) and (ii) of Theorem~\ref{thm:classic}, then stability of the $(G_{11},\Delta)$ interconnection is automatic since the $(1,1)$ block of the FDI provides the remaining requirement for stability in Theorem~\ref{thm:classic_performance}.
\end{rem}

Next, we'll need a way to convert an $\ltwo$ gain into an exponential rate bound. The sequel is similar to~\cite[Prop.~1]{megretski_system_1997}, but presented here with an explicit rate construction and adapted for discrete time systems.

\begin{lem}\label{lem:L2exp}
	Define the recursion with $x_0=0$ by:
	\begin{equation}\label{eq1}
		\begin{aligned}
			x_{k+1} &= \phi(x_k) + g_k\qquad k=0,1,2,\dots
		\end{aligned}
	\end{equation}
	where $\phi:\R^n\to\R^n$ satisfies $\phi(0)=0$. 
	Suppose that there exists a constant $c > 0$ such that whenever $g \in \ltwo$, and $(g,x)$ is a valid trajectory of~\eqref{eq1}, then
	\begin{equation}\label{eq2}
		\| x \|_{\ltwo}^2 \leq c\, \| g \|_{\ltwo}^2.
	\end{equation}
	Then, we also have the bound
	\[
	\|x_{k+1}\|^2_2 \leq \sum_{i=0}^{k} c\left(1-\frac1{c}\right)^{k-i} \!\!\|g_i\|^2_2\,.
	\]
\end{lem}
\begin{proof}
	We write $(x,g)\in\mathcal{S}$ to denote a valid trajectory of~\eqref{eq1}.
	Define the function $V:\R^n\to\R$ as follows:
	\[
	V(\xi) \defeq \sup_{\substack{{x_1}=\xi \\ g\in\ltwo \\ (g,x)\in\mathcal{S}}}
	\left( \|x\|_{\ltwo}^2 - c\, \|g\|_{\ltwo}^2 + c\, \|\xi\|_2^2 \right)\:.
	\]
	The first step is to bound this function. Note that because $x_0=0$ and $\phi(0)=0$, we have $\xi = x_1 = g_0$. An easy lower bound is found by specializing to $g_1=g_2=\dots=0$. An upper bound is found by using~\eqref{eq2}. The result is that
	\begin{equation}\label{bnd}
		\|\xi\|_2^2 \le V(\xi) \le c\,\|\xi\|_2^2\:.
	\end{equation}
	Fix $(\bar g,\bar x)\in\mathcal{S}$ to be any feasible trajectory of~\eqref{eq1}. We may lower-bound $V(\bar x_1)$ by setting $g_1 = \bar g_1$ and shifting the entire $x$ and $g$ vectors forward one timestep:
	\begin{align*}
		V(\bar x_1) &\ge \sup_{\substack{{x_1}=\bar x_1 \\ g\in\ltwo,\,g_1 = \bar g_1 \\ (g,x)\in\mathcal{S}}}
		\left( \|x\|_{\ltwo}^2 - c\, \|g\|_{\ltwo}^2 + c\, \|\bar x_1\|_2^2 \right) \\
		&= V(\bar x_2) + \|\bar x_1\|_2^2 - c\, \|\bar g_1\|_2^2 \\
		& \ge V(\bar x_2) + \tfrac{1}{c}\, V(\bar x_1) - c\,\|\bar g_1\|_2^2\:,
	\end{align*}
	where we made use of the bound~\eqref{bnd} in the final step. Rearranging, we obtain
	\[
	V(\bar x_2) \le \left(1-\frac1{c}\right) V(\bar x_1) + c\, \|\bar g_1\|^2_2\:.
	\]
	We may lower-bound $V(\bar x_3)$ by setting $g_1 = \bar g_2$ and using a similar argument. Continuing in this fashion,
	\[
	V(\bar x_{k+1}) \le \left(1-\frac1{c}\right) V(\bar x_k) + c\, \|\bar g_k\|^2
	\quad\text{for }k=0,1,2,\dots\:.
	\]
	It follows that for all $k$, we have
	\[
	V(\bar x_{k+1}) \le \left(1-\frac1{c}\right)^k V(\bar x_1) + \sum_{i=1}^{k} c\left(1-\frac1{c}\right)^{k-i} \|\bar g_i\|^2\:.
	\]
	Applying the bound~\eqref{bnd} one more time, we conclude that
	\begin{align*}
		\|\bar x_{k+1}\|^2 &\le V(\bar x_{k+1}) \\
		&\le \left(1-\frac1{c}\right)^k V(\bar x_1) +\sum_{i=1}^{k} c \left(1-\frac1{c}\right)^{k-i}  \|\bar g_i\|^2 \\
		&\le c\left(1-\frac1{c}\right)^k \|\bar x_1\|^2 +\sum_{i=1}^{k} c \left(1-\frac1{c}\right)^{k-i}  \|\bar g_i\|^2 \\
		&= \sum_{i=0}^{k} c \left(1-\frac1{c}\right)^{k-i} \|\bar g_i\|^2 
	\end{align*}
	where we used in the last step that $\bar g_0 = \bar x_1$.
	This completes the proof.
\end{proof}

By combining Theorem~\ref{thm:classic_performance} and Lemma~\ref{lem:L2exp}, we can find exponential rate bounds for LTI systems in feedback with nonlinearities that satisfy IQCs.
First, use the setup of Fig.~\ref{fig:lure_aug} with $d=g$ and $e=x$. Then, the $\ltwo$ bound in~\eqref{eq2} is an IQC as in~\eqref{iqcperf}, with
\[
\Pi_p = \begin{bmatrix} c & 0 \\ 0 & -1 \end{bmatrix}\:.
\]
Then, transform Fig.~\ref{fig:lure1} into augmented form by setting
\[
\begin{bmatrix}
G_{11} & G_{12} \\ G_{21} & G_{22}
\end{bmatrix}
=
\left[\begin{array}{c|cc}
A & B & B \\ \hlinet
C & D & D \\
I & 0 & 0
\end{array}\right]\:.
\]
Finally, the appropriate initial condition can be set by using $g=d=\begin{bmatrix} x_0^\tp & 0 & 0 & \dots \end{bmatrix}^\tp$. Applying Lemma~\ref{lem:L2exp} leads to a bound of the form $\|x_{k+1}\|_2^2 \le c\left(1-\tfrac{1}{c}\right)^k \|x_0\|_2^2$. Or, put another way, an exponential rate of $\rho = \sqrt{1-\tfrac{1}{c}}$.

The FDI of Theorem~\ref{thm:classic_performance} can be transformed into an LMI in a manner similar to that described in Section~\ref{sec:computation}. This LMI is linear in $P$ and $c$, so it can be efficiently solved to find the minimal $c$. This in turn allows us to find the smallest exponential rate $\rho$.

%!TEX root = boczar_submitted.tex

\section{IQC Library}\label{sec:library}
In this section, we show some classes of nonlinearities that can be described by $\rho$-IQCs and therefore used in Theorem~\ref{thm:main} to prove robust exponential stability of an interconnected system. In the case where $\rho=1$, these $\rho$-IQCs reduce to standard IQCs~\cite{megretski_system_1997}. This class of IQCs will be constructed for single-input single-output systems, but they may be adapted for square multi-input multi-output systems where the nonlinearity is of the form $\diag(\{\Delta_i\})$ for a scalar $\Delta$.

\subsection{Noisy Multiplication}
As noted for continuous time in \cite{huseiler}, nonlinearities of the form $\Delta(y_k) \equiv \delta_k y_k$ for some unknown and/or time-varying $\delta_k$ may satisfy $\rho$-IQCs. As $\Delta$ and $\rho_{\pm}$ commute, in the parlance of Prop. \ref{prop:iqcequiv} we have that $\Delta=\Delta'$, so $\Delta\in\IQC(\Pi,1)$ implies $\Delta\in\IQC(\Pi,\rho)$. See \cite{megretski_system_1997} for examples of IQCs for noisy multiplication.

\subsection{Uncertain Time Delay}%
The following is a discrete-time analog of the $\rho$-IQC first developed in \cite{huseiler}.
Let $\Delta$ be the operator defined by 
\begin{align*}
\Delta(y_k) = &\; \begin{cases}
0, & k<\tau\\
y_{k-\tau}, & k \geq \tau
\end{cases}\:,
\end{align*}
for some unknown $\tau$ in $[0,\tau_0]$, where $\tau_0$ is known. Now, observe that
\begin{align*}
\Delta'(y_k) = \rho^{-k}\Delta(\rho^{-k}y_k) = &\;\rho^{-k} \cdot \begin{cases}
0, & k<\tau\\
\rho^{-(k-\tau)}y_{k-\tau}, & k \geq \tau
\end{cases}\\
= &\; \rho^{-\tau}\Delta(y_k)\:.
\end{align*}
Thus, we may transform the system into one with a block diagonal nonlinearity $\diag\{\Delta, \rho^{-\tau}\}$. We can then use existing IQCs for noisy multiplication and time delays, always using $\Pi(\rho z)$ instead of $\Pi(z)$ \cite{huseiler}. 

Alternatively, with any bounded Hermitian function $X(\rho z) = X(\rho z)^*\succeq 0$, we see that 
\begin{align*}
&\bmat{\hat y(\rho z)\\\hat u(\rho z)}^*
\bmat{\rho^{-2\tau_0}X(\rho z) & 0\\0&-X(\rho z)}
\bmat{\hat y(\rho z)\\\hat u(\rho z)} \\
&\quad =
\bmat{\hat y(\rho z)\\\rho^{-\tau}\hat y(\rho z)}^*
\bmat{\rho^{-2\tau_0}X(\rho z) & 0\\0&-X(\rho z)}
\bmat{\hat y(\rho z)\\\rho^{-\tau}\hat y(\rho z)} \\
&\quad = (\rho^{-2\tau_0} - \rho^{-2\tau})\hat y(\rho z)^*X(\rho z)\hat y(\rho z) \geq 0\:.
\end{align*}
Thus, $\Delta \in \IQC(\diag\{\rho^{-2\tau_0}X(z),-X(z)\},\rho)$.

\subsection{Pointwise IQCs}
A nonlinearity $\Delta$ satisfies a pointwise IQC with a factorization $(\Psi, M)$ if $z_k^\tp Mz_k \geq 0$ for each $k$. In other words, the IQC holds pointwise in time. In this case, $\Delta$ also satisfies the associated $\rho$-IQC for all $\rho < 1$. 
Examples of pointwise IQCs include the $\gamma$ \emph{norm-bounded IQC}
\begin{equation*}
  \Pi = \begin{bmatrix}
  \gamma^2&0\\0&-1
  \end{bmatrix}\:,
\end{equation*}
and the \emph{$[\alpha,\beta]$ sector-bounded IQC}, given by
\begin{equation*}
  \Pi = \begin{bmatrix}
  -2\alpha\beta&\alpha+\beta\\\alpha+\beta&-2
  \end{bmatrix}\:,
\end{equation*}
which corresponds to nonlinearities $\Delta$ that satisfy
\begin{align*}
(\Delta(x)-\beta x)^\tp(\Delta(x)-\alpha x)\leq 0 \quad\forall \:x\:.
\end{align*}
Note that the norm-bounded IQC is a special case of the sector IQC with the sector $[-\gamma,\gamma]$. These IQCs hold even if $\Delta$ is time-varying, if $\Delta$ satisfies the IQC at each $k$.

\subsection{Zames--Falb IQCs}

A nonlinearity $\Delta$ is \emph{slope-restricted on $[\alpha,\beta]$} where \mbox{$0 \le \alpha \le \beta \le \infty$}
if the following relation holds for all $x$, $y$.
\begin{equation*}
\bl( \Delta(x)-\Delta(y)-\alpha(x-y) \br)^\tp
\bl( \Delta(x)-\Delta(y)-\beta(x-y) \br) \leq 0\:.
\end{equation*}
This relation states that the chord joining input-output pairs of $\Delta$ has a slope that is bounded between $\alpha$ and $\beta$.
This class of functions satisfies the Zames--Falb family of IQCs \cite{heath_zames-falb_2005,zames_stability_1968}. We give the definition below.

\begin{prop}\label{prop:zf}
A nonlinearity $\Delta$ that is static and slope-restricted on $[\alpha,\beta]$\footnote{The $\beta=\infty$ case for this and similar IQCs considers only the $\beta$ terms, i.e. $\Pi_{[\alpha,\infty]} = \lim_{\beta\to\infty}\beta^{-1}\Pi$.}
 satisfies the Zames--Falb IQC
\begin{equation}\label{eq:zf}
  \Pi = \addtolength{\arraycolsep}{0mm}
  \bmat{-\alpha\beta(2\!-\!\hat h  \!-\! \hat h ^*) & \alpha (1\!-\!\hat h ) \!+\! \beta(1\!-\!\hat h ^*) \\
  \alpha (1\!-\!\hat h ^*) \!+\! \beta(1\!-\!\hat h ) &
  -(2\!-\!\hat h \!-\!\hat h ^*)}
\end{equation}
where $\hat h(z)$ is any proper transfer function with impulse response $h \defeq (h_0,h_1,\dots)$ that satisfies $||h||_1 \le 1$ and $h_k\geq 0$ for all $k$. If $\Delta$ is odd ($\Delta(-x) = -\Delta(x)$), then we may remove the constraint that $h_k\ge 0$ for all $k$.
\end{prop}
\begin{proof}
See for example~\cite{heath_zames-falb_2005}.
\end{proof}

\begin{rem}\label{rem:zf}
The Zames--Falb IQC~\eqref{eq:zf} admits the factorization
\[
\Psi = \bmat{ \beta(1 - \hat h) & -(1-\hat h) \\ -\alpha & 1 }
\quad\text{and}\quad
M = \bmat{0 & 1 \\ 1 & 0 }\:.
\]
\end{rem}
In general, for a given fixed $\rho$, only a subset of the Zames--Falb IQCs will be $\rho$-IQCs. We now give a characterization of this subset.

\begin{thm}[Zames--Falb $\rho$-IQC]\label{thm:rho_zf}
  Suppose $\Delta$ is static and slope-restricted on $[\alpha,\beta]$. Then $\Delta\in\IQC(\Pi(z),\rho)$ where $\Pi$ is the Zames--Falb IQC~\eqref{eq:zf} and $\hat h$ satisfies the additional constraint
  \begin{equation}\label{eq:rhozf_constraint}
  \sum_{k=0}^\infty \rho^{-2k} |h_k| \le 1 \:.
  \end{equation}
\end{thm}
\begin{proof}
The proof involves rewriting the IQC as a discrete-time sum which can be split into parts that can separately be shown to be nonnegative. See Appendix~\ref{sec:zfthmproof} for the full proof of Theorem~\ref{thm:rho_zf} and related extensions.
\end{proof}
Sector-bounded and/or slope-restricted functions show up in various specialized contexts. We will derive $\rho$-IQCs for two such cases: stiction nonlinearities and quasi-monotone/quasi-odd nonlinearities.

\subsubsection{Stiction Nonlinearities}
Stiction nonlinearities (shown in Fig.~\ref{fig:stiction}) satisfy Zames--Falb $\rho$-IQCs with additional constraints on the coefficients $h_k$.

\begin{figure}[thpb]
  \centering
  \includegraphics[scale=0.9]{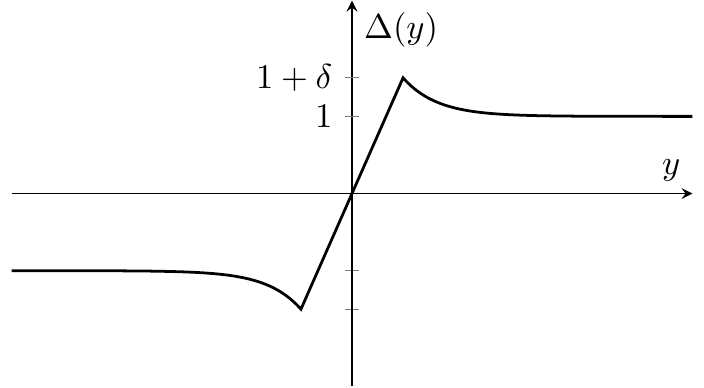}
  \caption{Example stiction nonlinearity (taken from \cite{rantzer_stiction}.)}
  \label{fig:stiction}
\end{figure}
\begin{cor}[Stiction $\rho$-IQC]\label{thm:stiction}
  Suppose $\Delta$ is a stiction nonlinearity with slope $1/\epsilon$ and overshoot $\delta$ as defined in \cite{rantzer_stiction}. Then $\Delta\in\IQC(\Pi(z),\rho)$ where $\Pi$ is the $[0,1/\epsilon]$ Zames--Falb IQC~\eqref{eq:zf} and $H$ satisfies the additional constraint
  \begin{equation*}
  \sum_{k=0}^\infty \rho^{-2k} |h_k| \le \frac{1-\delta}{1+\delta}\:.
  \end{equation*}
\end{cor}

\subsubsection{Quasi-monotone and Quasi-odd Nonlinearities}
Following the definition in \cite{heath_genzf} (shown in Fig.~\ref{fig:quasimonotone}), quasi-monotone and quasi-odd
nonlinearities also satisfy Zames--Falb $\rho$-IQCs under additional constraints on the $h_k$.

\begin{figure}[thpb]
  \centering
  \includegraphics[scale=1.0]{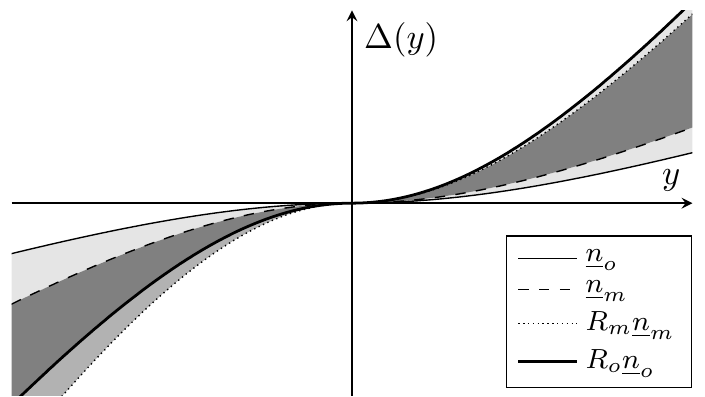}
  \caption{Monotone and odd bounds for unknown nonlinearities (modified from \cite{heath_genzf}). The nonlinearity must lie within envelopes generated by multiplicative perturbations of a known monotone linearity $\underbar{$n$}_m$ (perturbation between $1$ and $R_m\geq 1$) and a known monotone odd nonlinearity $\underbar{$n$}_o$ (perturbation between $1$ and $R_o\geq 1$). In this example, the nonlinearities of interest lie in the darkest region, the intersection of both envelopes.}
  \label{fig:quasimonotone}
\end{figure}

\begin{cor}[Quasi-monotone/odd $\rho$-IQC]\label{thm:qmp_zf}
  Suppose $\Delta$ is static and is quasi-monotone or quasi-odd as defined in \cite{heath_genzf}. Then $\Delta\in\IQC(\Pi(z),\rho)$ where $\Pi$ is the Zames--Falb IQC~\eqref{eq:zf} and $H$ satisfies the additional constraint
  \begin{equation*}
  \sum_{k=0}^\infty \gamma_k^{-1} \rho^{-2k}|h_k| \le 1
  \end{equation*}
  where
  \begin{equation*}
    \gamma_k^{-1} \defeq \begin{cases}
    R_m, & \: h_k \geq 0 \\
    R_o, & \: h_k < 0
    \end{cases}\:.
  \end{equation*}
\end{cor}
Given a fixed $\rho$, searching over (finite) $h_k$ when solving the feasibility LMI using this IQC is still a convex problem.  To see this, observe that we can equivalently write this constraint on the $h_k$ (assuming $R_m \geq R_o$, the other case is similar) as
\begin{align*}
R_o \sum_{k=0}^K \rho^{-2k}|h_k| + (R_m-R_o)\sum_{k=0}^K \rho^{-2k}\cdot\max(h_k, 0) \le 1 \:.
\end{align*}
However, the proof of Corollary~\ref{thm:qmp_zf} will show that these general Zames--Falb $\rho$-IQCs can be written as a nonnegative linear combination off ``off-by-$j$'' $\rho$-IQCs. Thus, when solving \eqref{eq:explmi} it is  sufficient to search over all nonnegative linear combinations of simpler $\rho$-IQCs atoms, rather than formulating the constraint on the $h_k$ explicity. Whether this is more efficient depends on the specific problem dimensions.
The correct chain of implications for this constraint (and others) is as follows:

\begin{itemize}
  \item Compared to the odd Zames--Falb IQC, a quasi-odd IQC as defined in Corollary~\ref{thm:qmp_zf} gives \emph{less information} about the nonlinearity $\phi$, i.e. we must provide a certificate of stability for every nonlinearity in a \emph{larger class}.
  \item Since $R_m, R_o \geq 1$, the weights satisfy $\gamma_k^{-1}\ge 1$, so there is \emph{less freedom} in choosing the $h_k$.
  \item This restriction in choosing $h_k$ leads to a \emph{smaller} feasible set for the LMI. Thus, the upper bound we find for the convergence rate will be \emph{larger}.
\end{itemize}

\subsection{Repeated Sector Nonlinearities}\label{sec:repeated_nonlinearities}

We say a real symmetric matrix $\Gamma$ is \emph{$(\rho,H)$-diagonally dominant} if, for a symmetric matrix of nonnegative proper transfers functions $\hat H$ with impulse responses $H_{ij,k}$, we have that $\Gamma_{ii}\geq 0$, $\Gamma_{ij}\leq 0$ (for $i\ne j$), $H_{ij,k}\geq 0$, $\sum_{k=0}^\infty\rho^{-2k}|H_{ij,k}| \leq 1 \: \forall\: (i,j)$ and
\begin{align*}
\Gamma_{ii} \geq \sum_{j=1, j\neq i}^n |\Gamma_{ij}| + \sum_{j=1}^n\sum_{k=0}^\infty\rho^{-2k}|H_{ij,k}| \quad \forall\:i\:.
\end{align*}
We call $\Gamma$ simply \emph{diagonally dominant}\footnote{Note that the conventional definition of ``diagonally dominant'' does not restrict the diagonal elements to be nonnegative.} if the above holds with $H=0$ and $\rho=1$.

Now, let $\Delta$ be a repeated monotone scalar nonlinearity in some sector, i.e. $\Delta(y) = \diag\{\phi(y_i)\}$.

\begin{prop}\label{prop:constant_dd}
$\Delta$ satisfies the pointwise $\rho$-IQC
\begin{equation*}
\Pi =\; \bmat{0 & \Gamma \\ \Gamma & 0}
\end{equation*}
for any symmetric diagonally dominant matrix $\Gamma$.
\end{prop}
\begin{proof}
The proof is analogous to the proof of Theorem 1 in the Appendix of \cite{damato_repeated} with $H=0$.
\end{proof}

\begin{thm}\label{thm:rhozf_dd}
Assume $\Gamma$ is $(\rho,H)$-diagonally dominant. Then, if $\phi$ is in the $[\alpha, \beta]$ sector, then $\Delta(y) = \diag\{\phi(y_i)\}$ satisfies the $\rho$-IQC
\begin{equation}\label{eq:pi_ddzf}
\Pi =\; \bmat{-\alpha\beta(2\Gamma-\hat H-\hat H^*) & \alpha(\Gamma-\hat H)+\beta(\Gamma-\hat H^*) \\ \alpha(\Gamma-\hat H^*)+\beta(\Gamma-\hat H) & -2\Gamma+\hat H+\hat H^*}.
\end{equation}
\end{thm}
\begin{proof}
The proof is similar in spirit to that of Theorem \ref{thm:rho_zf} but more involved; see Appendix~\ref{sec:ddthmproof}.
\end{proof}

\begin{rem}\label{rem:Gammafact}
The repeated $[\alpha,\beta]$-sector nonlinearity $\rho$-IQC admits the factorization
\begin{align*}
\Psi = &\; \bmat{\beta (\Gamma-\hat H) & -(\Gamma-\hat H) \\ -\alpha I &  I},\quad
M = \bmat{0 & I \\ I & 0}\:.
\end{align*}
\end{rem}
See Appendix~\ref{sec:appendixcomp} for a note on how to search over general nonnegative combinations of $\rho$-IQCs of the form (\ref{eq:pi_ddzf}), which is not immediately apparent.

%!TEX root = boczar_submitted.tex

\section{Examples}\label{sec:examples}
\subsection{Using multiple IQCs}
Using multiple IQCs can lead to a more refined $L_2$ gain bound. Likewise, using multiple $\rho$-IQCs can lead to refined exponential rates. In this section, we present numerical examples using both pointwise and dynamic $\rho$-IQCs.

Consider a stable discrete-time LTI system $G(z)$ in feedback with the sigmoidal nonlinearity $\Delta(x) = b\arctan(x)$. This interconnection is shown in Fig.~\ref{fig:lure_ex}.
\begin{figure}[th]
  \centering
  \includegraphics{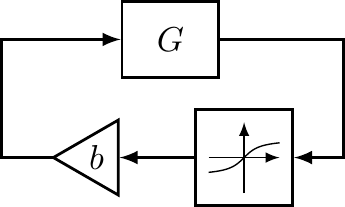}
  \caption{LTI system $G$ in feedback with the static sigmoidal nonlinearity \mbox{$\Delta(x)=b \arctan(x)$}.}\label{fig:lure_ex}
\end{figure}

Since this nonlinearity is static, in the $[0,b]$ sector, and $[0,b]$ slope-restricted, it satisfies the following $\rho$-IQCs:
\begin{align}
\Pi_\textup{n}(z) &\triangleq \bmat{b^2&0\\0&-1}  \quad \text{(norm-bounded)} \label{aa}\\
\Pi_0(z) &\triangleq \bmat{0&b\\b&-2} \quad \text{(sector bounded)} \label{bb}\\
\Pi_k(z) &\triangleq
\bmat{0&b(1 - \rho^{2k} \bar z^{-k}) \\ b(1 - \rho^{2k} z^{-k}) & -2 +\rho^{2k}(z^{-k}+\bar z^{-k})} \label{cc}\\
&\; \quad \text{(off-by-$k$ Zames--Falb)} \nonumber
\end{align}
where we may choose any $k\geq 1$. 

\paragraph*{A simple bound.} For our first case study, we analyzed the interconnection of Fig.~\ref{fig:lure_ex} with the LTI system\footnote{This example was inspired by the continuous-time example given in \cite{scherer}, which showed that adding more IQCs yields better $L_2$ gain bounds.}
\begin{equation}\label{G1}
G_1(z) = -\frac{(z+1)(10z+9)}{(2z-1)(5z-1)(10z-1)}\:.
\end{equation}
We solved the feasibility LMI~\eqref{eq:explmi} using MATLAB together with CVX~\cite{cvx2,cvx} to find the fastest guaranteed rate of convergence and we searched over positive linear combinations of subsets of the IQCs~\eqref{aa}--\eqref{cc}. Fig.~\ref{fig:rates} shows the rate bounds achieved as a function of which IQCs were used. Fig.~\ref{fig:states} shows sample state trajectories for the case $b=1$.

The true exponential rate can be found by linearizing the system about its equilibrium point. Namely, $\Delta(x) \approx bx$. Formally, this is an application of Lyapunov's indirect method~\cite[Thm.~4.13]{khalil}. The result is that the decay rate should correspond to the maximal pole magnitude of the closed-loop map $G(z)/(1-bG(z))$. We display the true exponential rate as the dashed black curve in Fig.~\ref{fig:rates} and Fig.~\ref{fig:states}.

For this example, the $\rho$-IQC approach yields a tight upper bound to the true exponential rate when we use a combination of the sector and off-by-1 IQCs. We also computed the exponential rate derived from $\ell_2$ gain as described in Section~\ref{sec:L2exp} (dotted line). The $\ell_2$ bound is very conservative despite being computed using all available IQCs.

\begin{figure}[th]
  \centering
  \includegraphics{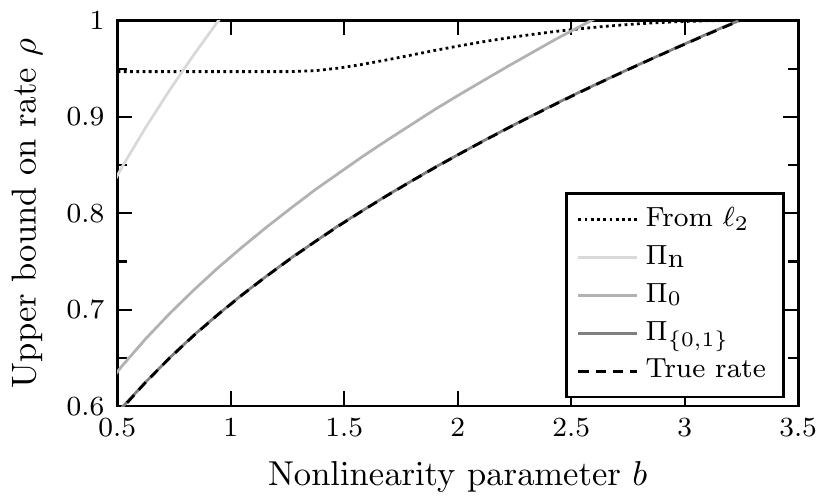}
  \caption{Upper bounds on the exponential convergence rate $\rho$ for the system $G_1(z)$ given in~\eqref{G1} in feedback as in Fig.~\ref{fig:lure_ex}. A tight bound is achieved using two $\rho$-IQCs. The bound derived from the $\ell_2$ gain is very conservative.\label{fig:rates}}
\end{figure}

\begin{figure}[th]
  \centering
  \includegraphics{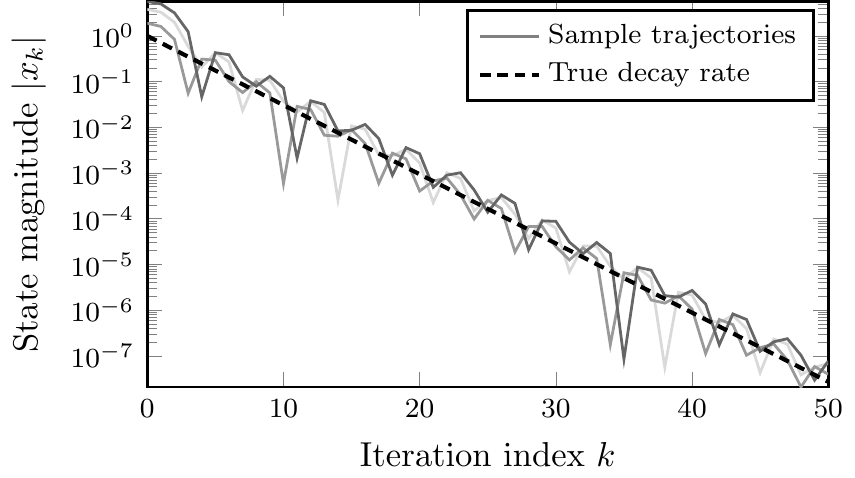}
  \caption{State decay over time of the system $G_1(z)$ in feedback as in Fig.~\ref{fig:lure_ex} with $b=1$ for various initial conditions $x_0\in[-15,15]$. The dashed black line is $\rho^k$, where $\rho=.7058$ is the true rate at $b=1$ in Fig.~\ref{fig:rates}.}
  \label{fig:states}
\end{figure}

\paragraph*{A more complex bound.} The $\rho$-IQC approach does not always achieve tight bounds as in the previous example. Consider the same interconnection of Fig.~\ref{fig:lure_ex} but this time using
\begin{equation}\label{G2}
G_2(z) = \frac{2z-1}{10(2z^2-z+1)}
\end{equation}
The rate bounds for various $\rho$-IQCs are shown in Fig.~\ref{fig:rates_ex2}. This time, we again observe that using more IQCs achieves better rate bounds, but the bound is not tight even after using six IQCs. However, if we add the Zames--Falb IQCs corresponding to odd monotone nonlinearities, the rate improves to within a small tolerance of the true rate.

\begin{figure}[thpb]
  \centering{}
  \includegraphics{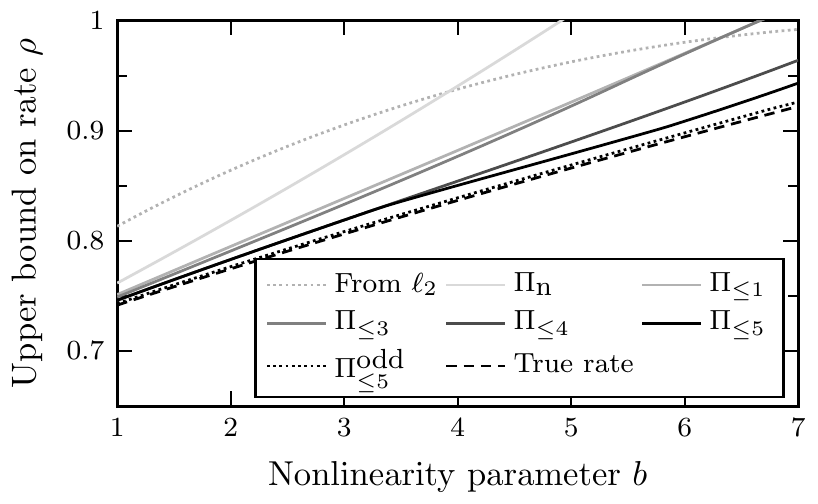}
  \caption{ Upper bounds on the exponential convergence rate $\rho$ for the system $G_2(z)$ given in~\eqref{G2} in feedback as in Fig.~\ref{fig:lure_ex}. As we include more $\rho$-IQCs, we can certify tighter bounds. Once again, the $\ell_2$-derived bound is more conservative.}
  \label{fig:rates_ex2}
\end{figure}

As in the previous example, the best achievable rate derived from an $\ell_2$ gain bound as detailed in Section~\ref{sec:L2exp} is still very conservative when compared to the rates obtained by using the $\rho$-IQC approach.

\paragraph*{A quasi-odd nonlinearity} Consider the asymmetric nonlinearity in Fig.~\ref{fig:asym_nonlin}, shown with the associated monotone and odd bounds as defined in \cite{heath_genzf}. In this example, we have $R_m = 1$ and $R_o = 2$.
\begin{figure}[thpb]
  \centering{}
  \includegraphics{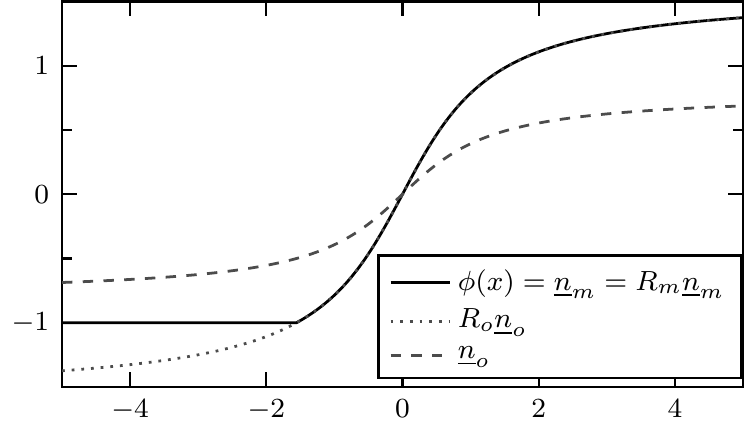}
  \caption{Plot of the monotone and quasi-odd asymmetric nonlinearity \mbox{$\phi(x) = \max\{\arctan(x),-1\}$} with its associated bounds.}
  \label{fig:asym_nonlin}
\end{figure}
Thus, we may invoke Corollary~\ref{thm:qmp_zf} and use the associated $\rho$-IQC. Using this system in feedback with the $G(z)$ from the second example, we see in Fig.~\ref{fig:quasi_zf} that the quasi-odd Zames--Falb IQCs yield better performance than the monotone Zames--Falb IQCs of the same order (which requires all filter coefficients $h_k$ to be positive).

\begin{figure}[thpb]
  \centering{}
  \includegraphics{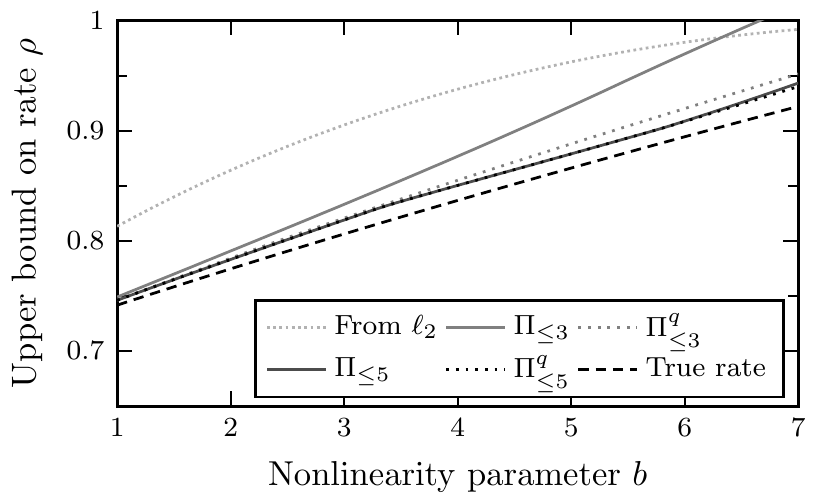}
  \caption{Comparison of monotone Zames--Falb and quasi-odd (denoted with superscript $q$) Zames--Falb IQC rate certificates.}
  \label{fig:quasi_zf}
\end{figure}

\paragraph*{Repeated nonlinearities}
To illustrate the need for repeated nonlinearity IQCs, first instantiate some stable SISO system $G$ with realization $(A,B,C,D)$. Now, consider the ``extended'' 2-input 2-output system
\begin{align*}
G_{\text{ext}} = &\; 
\left[\begin{array}{c|cc}
A & B & -B \\ \hlinet
C & D & 0\\
C & 0 & D
\end{array}\right]
\end{align*}
and connect this system in positive feedback with the block-diagonal nonlinearity \mbox{$\Delta=\diag\{\Delta_1,\Delta_2\}$}. If we constrain $\Delta_1=\Delta_2$, then the nonlinearities cancel each other out and the system is in open loop. The convergence rate of the state is therefore determined by the largest magnitude eigenvalue of $A$. However, if our IQC does not capture that the nonlinearity is repeated and instead only assumes each individual nonlinearity is (say) $[0,b]$-slope restricted, then $G$ must essentially be robust to $b$-norm bounded nonlinearities in the feedback loop. This will result in a worse rate certificate or even none at all (if $G$ is made unstable by positive feedback). 

Indeed, constructing $G_{\text{ext}}$ using our previous ``tight bound'' example with $b=0.3$ leads to a rate certificate of $\approx 0.825$ using only the odd monotone IQC; replacing it with the repeated odd monotone nonlinearity IQC  gives a certificate matching
the true convergence rate, $0.5$.
%!TEX root = boczar_submitted.tex

\section{Conclusion}\label{sec:conclusion}

\LL{IQC theory is the most general tool available for certifying robust stability of systems in feedback with unknown, uncertain, or otherwise difficult nonlinearities. As stable systems are often exponentially stable, it is reasonable to want finer control over not only stability, but also exponential decay rate.}

\LL{The generalization presented herein enables the certification of robust \textit{exponential} stability with precise control over the decay rate. Moreover, the library of $\rho$-IQCs provided shows how this approach can be applied as broadly and efficiently as the classical IQC theory.}

\if\MODE1
\section{Acknowledgments}\label{sec:acknowledgments}
{}
The authors would like to thank \mbox{Murat Arcak} for very helpful discussions. L.~Lessard was partially supported by AFOSR award FA9550-12-1-0339. R.~Boczar is supported by the Department of Defense NDSEG Scholarship.
\fi
%!TEX root = boczar_submitted.tex

\appendix

\section{Appendix}

\subsection{Proof of Proposition~\ref{prop:exp}}\label{sec:prop5proof}

Suppose the interconnection of Fig.~\ref{fig:lure2} is stable. Then there exists some $K>0$ such that for any choice of the signals $e$ and $f$ and for all $T$,
\begin{equation}\label{hh}
\sum_{k=0}^T  \bl(\|w_k\|^2 + \|v_k\|^2\br) \leq  K \sum_{k=0}^T  \bl(\|e_k\|^2 + \|f_k\|^2\br)\:.
\end{equation}
The proof will follow by carefully choosing $e$ and $f$ to transform Fig.~\ref{fig:lure2} into Fig.~\ref{fig:lure1}.
To this end, note that $(A,B)$ is controllable by assumption. So there exists a finite sequence of inputs $u_0,\dots,u_{n-1}$ and corresponding outputs $y_0,\dots,y_{n-1}$ that drives the state of $G$ from $\xi_0=0$ to $\xi_n=x_0$. Therefore, if we set
\[
e_k = \begin{cases}\rho^{-k}u_k & 0\le k < n\\ 0 & k \ge n\end{cases}
,\;\;\;
f_k = \begin{cases}-\rho^{-k}y_k & 0\le k < n\\ 0 & k \ge n\end{cases}
\]
then we obtain $\xi_n = x_0$ in the interconnection of Fig.~\ref{fig:lure2}. Moreover, $\rho_-\circ\rho_+$ is the identity operator. It follows that for  $k\ge n$, the two interconnections become identical and therefore $\xi_k = x_{k-n}$.

Substituting into~\eqref{hh}, we conclude that
\begin{equation}\label{hhh}
\sum_{k=0}^T  \bl(\|w_k\|^2 + \|v_k\|^2\br) \leq  K \sum_{k=0}^{n-1} \bl(\|e_k\|^2 + \|f_k\|^2\br)\:.
\end{equation}
The right-hand side of~\eqref{hhh} is independent of $T$, but~\eqref{hhh} holds for all $T$ so we must have
\begin{equation*}
\lim_{k \to\infty} \|w_k\| = 0
\qquad\text{and}\qquad
\lim_{k \to\infty} \|v_k\| = 0\:.
\end{equation*}
For $k\ge n$, we have $w_k = \rho^{-k}u_k$ and $v_k = \rho^{-k} y_k$. Therefore there exists some constant $c>0$ such that
\[
\|u_k\| \le c \rho^k
\qquad\text{and}\qquad
\|y_k\| \le c \rho^k \:.
\]
Now $(A,C)$ is observable by assumption, so let $L$ be such that the eigenvalues of $A+LC$ are all zero. Rewrite the dynamics of $G$ as
\begin{equation}\label{gg}
x_{k+1} = \bar A x_k +\bar B h_k
\end{equation}
where $\bar A \defeq A+LC$, $\bar B \defeq  \bmat{LD+B & -L}$, and $h_k \defeq \bmat{u_k^\tp  & y_k^\tp }^\tp $.
Iterating~\eqref{gg}, we obtain
\begin{equation}\label{ggg}
x_k = \bar A^k x_0 + \sum_{i=0}^{k-1} \bar A^{k-1-i} \bar B h_i\:.
\end{equation}
Since all eigenvalues of $\bar A$ are zero, $\bar A$ is nilpotent and so $\bar A^n = 0$. For $k\ge n$, \eqref{ggg} therefore becomes
\begin{equation*}\label{qq}
x_k = \sum_{i=0}^{n-1} \bar A^{n-1-i} \bar B h_{k-n+i}\:.
\end{equation*}
We can now bound the state using the triangle inequality.
\begin{align*}
\|x_k\|
&\le \underbrace{\bl\| \bmat{ \bar A^{n-1}\bar B & \dots & \bar A \bar B & \bar B} \br\|}_\gamma \sum_{i=k-n}^{k-1} \|h_i\| \\
&\le 2\gamma\, c \left( \frac{ \rho^{-n}-1}{1-\rho} \right) \rho^k\:,
\end{align*}
and this completes the proof.
\qedhere

\subsection{Proof of Theorem~\ref{thm:rho_zf} and related extensions}\label{sec:zfthmproof}
\subsubsection{\texorpdfstring{$[0,\infty]$-slope restricted case}{[0, Inf]-slope restricted case}}\label{sec:thm20}

We will prove this general result by first considering the simpler case where the slope restriction is on $[\alpha,\beta] = [0,\infty]$ and $H(z) = \pm\gamma_l\rho^{2j} z^{-l}$ for some constants  $0<\gamma_l\leq 1$. Note that this choice trivially satisfies~\eqref{eq:rhozf_constraint}, and the extensions of \eqref{eq:rhozf_constraint} follow for specific restrictions on $\gamma_k$ and mixed-sign $h_k$. In this case, the~$\Pi$ from~\eqref{eq:zf} (first taking the positive sign in $H(z)$) becomes
\begin{equation}\label{obj}
\Pi = \bmat{ 0 & 1-\gamma_j\rho^{2j}\bar z^{-l} \\ 1-\gamma_l\rho^{2l}z^{-l} & 0 }
\end{equation}
where $\bar z$ denotes the complex conjugate of $z$. We call~\eqref{obj} the ``off-by-$l$'' Zames--Falb IQC. We would like to show that $\Delta \in \IQC(\Pi(z),\rho)$. Appealing to Definition~\ref{def:piqc} and Remarks~\ref{rem:rhoIQC_time_domain} and~\ref{rem:zf}, this amounts to proving that
\begin{equation}\label{kk2}
\sum_{k=0}^\infty \rho^{-2k} u_k^\tp ( y_k -\gamma_l\rho^{2l} y_{k-l} ) \ge 0\:.
\end{equation}
We will prove~\eqref{kk2} by borrowing the approach from~\cite{lessard_analysis_2014}. If $\Delta$ is multidimensional, we require that $\Delta$ be the gradient of a potential function~\cite{heath_zames-falb_2005}. By the assumption that $\Delta$ is slope-restricted on $[0,\infty]$, we have
\[
(\Delta(x)-\Delta(y))^\tp(x-y) \ge 0 \quad\text{holds for all }x,y \:.
\]
In other words, $\Delta$ is monotone. Now define the scalar function $g$ such that $\grad g = \Delta$. By Kachurovskii's theorem, $g$ is convex and we have
\[
g(y) \ge g(x) + \Delta(x)^\tp(y-x)
\quad\text{for all }x,y \:.
\]
Moreover, setting $(x,y)\mapsto(y_k,0)$ or $(x,y)\mapsto(y_k,y_{k-l})$ leads to the two inequalities:
\begin{align}
u_k^\tp y_k &\ge g(y_k) \label{f1} \\
u_k^\tp (y_k - y_{k-l}) &\ge g(y_k) - g(y_{k-l})\:. \label{f2}
\end{align}
We will assume for simplicity that $g(x) \ge 0$ for all $x$, and we will first prove the case where we take the positive sign in $H(z)$.
Substituting~\eqref{f1} and \eqref{f2} into the left-hand side of~\eqref{kk2}, the partial sum from $0$ to $T$ is:
\if\MODE1
\begin{align*}
& \sum_{k=0}^T \rho^{-2k} u_k^\tp ( y_k- \gamma_l \rho^{2l} y_{k-l} ) \\
&=\sum_{k=0}^T \rho^{-2k}\bl( (1-\gamma_l\rho^{2l})u_k^\tp y_k + \gamma_l\rho^{2l}u_k^\tp(y_k- y_{k-l})\br) \\
&\ge \sum_{k=0}^T \rho^{-2k}\bl( (1-\gamma_l\rho^{2l})g(y_k) + \gamma_l \rho^{2l}(g(y_k)-g(y_{k-l})) \br) \\
&= \sum_{k=0}^T \rho^{-2k}\bl( g(y_k) - \gamma_l\rho^{2l}g(y_{k-l}) \br) \\
&= \sum_{k=0}^{T-l}(1-\gamma_l) \rho^{-2k} g(y_k) + \sum_{k=T-l+1}^{T} \rho^{-2k} g(y_k) \ge 0\:.
\end{align*}
\else
\begin{align*}
\sum_{k=0}^T \rho^{-2k} u_k^\tp ( y_k- \gamma_l \rho^{2l} y_{k-l} )
&=\sum_{k=0}^T \rho^{-2k}\bl( (1-\gamma_l\rho^{2l})u_k^\tp y_k + \gamma_l\rho^{2l}u_k^\tp(y_k- y_{k-l})\br) \\
&\ge \sum_{k=0}^T \rho^{-2k}\bl( (1-\gamma_l\rho^{2l})g(y_k) + \gamma_l \rho^{2l}(g(y_k)-g(y_{k-l})) \br) \\
&= \sum_{k=0}^T \rho^{-2k}\bl( g(y_k) - \gamma_l\rho^{2l}g(y_{k-l}) \br) \\
&= \sum_{k=0}^{T-l}(1-\gamma_l) \rho^{-2k} g(y_k) + \sum_{k=T-l+1}^{T} \rho^{-2k} g(y_k) \ge 0\:.
\end{align*}
\fi
Since each partial sum is nonnegative, the infinite sum (which must converge) is also nonnegative, and therefore we have proven~\eqref{kk2}. Now, for the case where we take the negative sign in $H(z)$, further assume that $\Delta$ is an odd function, which implies $g$ is an even function. Thus, using this fact and convexity inequality for $g$ with $(x,y)\mapsto(y_k,-y_{k-l})$ leads to the additional inequality
\begin{align*}
u_k^\tp (y_k + y_{k-l}) &\ge g(y_k) - g(y_{k-l})\:. \label{f3}
\end{align*}
The proof of nonnegativity of the partial sums then follows as before. Thus, a $[0,\infty]$-slope restricted $\Delta$ satisfies the off-by-$l$ $\rho$-IQC (and also the negative version if $\Delta$ is assumed to be odd).

Now we consider the case of a more general $\hat h(z)$. Suppose $\hat h(z) = \sum_{k=0}^\infty h_k z^{-k}$ where $h_k$ satisfies $\sum_k\gamma_k^{-1}\rho^{-2k}|h_k|\leq 1$. Then,
\if\MODE1
\begin{align*}
1 &- \hat h(z) = \\
&\underbrace{1-\sum_{k=0}^\infty \gamma_k^{-1}\rho^{-2k}|h_k|}_{\equiv c}
+ \sum_{h_k\geq0}\gamma_k^{-1}\rho^{-2k}h_k\left(1-\gamma_k\rho^{2k}z^{-k}\right) \\
&\;+ \sum_{h_k<0}\gamma_k^{-1}\rho^{-2k}(-h_k)\left(1+\gamma_k\rho^{2k}z^{-k}\right) \\
= &\;  c\,(1-\hat h_s) + \sum_{h_k\geq 0}\gamma_k^{-1}\rho^{-2k}h_k\left(1-\hat h_k^+(z)\right) \\
&\;+ \sum_{h_k< 0}\gamma_k^{-1}\rho^{-2k}(-h_k)\left(1-\hat h_k^-(z)\right)\:,
\end{align*}
\else
\begin{align*}
1 - \hat h(z) &= \underbrace{1-\sum_{k=0}^\infty \gamma_k^{-1}\rho^{-2k}|h_k|}_{\equiv c}
+ \sum_{h_k\geq0}\gamma_k^{-1}\rho^{-2k}h_k\left(1-\gamma_k\rho^{2k}z^{-k}\right) \\
&\hspace{6cm}+ \sum_{h_k<0}\gamma_k^{-1}\rho^{-2k}(-h_k)\left(1+\gamma_k\rho^{2k}z^{-k}\right) \\
&= c\,(1-\hat h_s) + \sum_{h_k\geq 0}\gamma_k^{-1}\rho^{-2k}h_k\left(1-\hat h_k^+(z)\right)
+ \sum_{h_k< 0}\gamma_k^{-1}\rho^{-2k}(-h_k)\left(1-\hat h_k^-(z)\right),
\end{align*}
\fi
where $\hat h_k^\pm(z)=\pm\gamma_k\rho^{2k}z^{-k}$ and $\hat h_s=0$ (for illustration). Note that $\hat h_k(z)$ corresponds the off-by-$k$ Zames--Falb IQC, which we proved above is a $\rho$-IQC, where the negative version is only used (with corresponding negative $h_k)$ if $\Delta$ is assumed to be odd. Also, $1-\hat h_s$ corresponds to the sector IQC, which is also a $\rho$-IQC. Now note that the general Zames--Falb IQC~\eqref{eq:zf} is linear in $1-\hat h$ and $1-\hat h^*$. Therefore, since by assumption $c \geq 0$, $\Pi(z)$ is a positive linear combination of $\rho$-IQCs and must therefore be a $\rho$-IQC itself. \qedhere

\subsubsection{Specific Zames--Falb classes}
We would now like to generalize this proof (or equivalently, specify further the class of nonlinearities). Now, assume that the nonlinearity $\Delta$ can be written as $\Delta_2 \circ \Delta_1$, where $\Delta_1\in\IQC(\Pi(z),\rho)$ of the Zames--Falb type in the preceding section where \mbox{$\sum_k \gamma_k^{-1} \rho^{-2k}|h_k|\leq 1$}. Further assume that $\Delta_2$ (which is possibly time-varying) satisfies
\begin{align*}
(1-\delta)u_k\leq \Delta_2(u_k)\leq (1+\delta)u_k\quad \forall \:u_k,k
\end{align*}
or equivalently,
\begin{align*}
\Delta_2(u_k) = (1+\delta_k)u_k, \:|\delta_k|\leq\delta
\end{align*}
for some $\delta<1$. We would like to show under what conditions $\Delta_2 \circ \Delta_1$ satisfies a Zames--Falb $\IQC$ with rate $\rho$.

To do this, we will show that $\Delta_2 \circ \Delta_1$ satisfies the relevant off-by-$l$ Zames--Falb $\rho$-IQC, which then extends to general Zames--Falb by the preceding section.
As in Section~\ref{sec:thm20}, we would like to show that
\begin{equation*}\label{nzf}
\sum_{k=0}^\infty \rho^{-2k} \Delta(y_k)^\tp (y_k \mp \gamma_l\rho^{2l} y_{k-l} ) \ge 0\:.
\end{equation*}
Using our prescribed $\Delta$ (taking the negative sign for simplicity), we see that each partial sum satisfies
\begin{align*}
\sum_{k=0}^\T &\; \rho^{-2k} (1+\delta_k)u_k^\tp (y_k - \gamma_l\rho^{2l} y_{k-l}) \quad(\Delta_1(y_k)\equiv u_k)\\
\geq &\;\sum_{k=0}^T \rho^{-2k}(1+\delta_k)\bl( g(y_k) - \gamma_l\rho^{2l}g(y_{k-l}) \br)
\end{align*}
by the same argument from the preceding section. This is then equal to 
\begin{align*}
\sum_{k=0}^{T-l}(1+\delta_k-(1+\delta_{k+l})\gamma_l) \rho^{-2k} g(y_k)
+\sum_{k=T-l+1}^{T} (1+\delta_k)\rho^{-2k} g(y_k)\:.
\end{align*}
A sufficient condition for this sum to be positive is
\begin{align*}
\gamma_l \leq \frac{1+\delta_k}{1+\delta_{l+k}}\:\forall\:k,
\end{align*}
which is satisfied if, for example,
\begin{align*}
 \gamma_l \equiv \gamma\leq \frac{1-\delta}{1+\delta}\:.
\end{align*}
If so, the partial sums converge and so does the infinite sum. Again, if we further assume that $\Delta_1$ is odd, the negative off-by-$l$ IQC is also satisfied. The argument for general $H(z)$ follows as before.\qedhere

Proofs for specific classes of nonlinearities in the literature correspond to specific choices of $\gamma_k$ and $\delta_k$. These are summarized in Table \ref{tab:zf}.
\if\MODE1
\begin{table*}[!t]
\renewcommand{\arraystretch}{1.5}
\caption{Variable Choices for Specific Zames--Falb Proofs}
\label{table_example}
\centering
\begin{tabular}{cccc}
\hline
\bfseries Type & \bfseries $\delta_k$ & \bfseries $\gamma_k$ & \bfseries Notes\\
\hline\hline
$[0,\infty]$-slope restricted & $0$ & $1$ & \\
$[\alpha,\beta]$-slope restricted & $0$ & $1$ &\shortstack{Loop transformation \\$(y,u) \mapsto (\beta y - u,u-\alpha y)$}\\
Noisy composition & $\delta$ & $1$ & \\
Stiction (slope $1/\epsilon$) &$\delta$ & $\dfrac{1-\delta}{1+\delta}$ &\\
Quasi-monotone/odd &$\dfrac{R_m-1}{R_m+1}, \dfrac{R_o-1}{R_o+1}$ & $R_m^{-1}, h_k\geq 0; \:R_o^{-1},h_k<0$ &Notational change in def. from \cite{heath_genzf}: $\underbar{$n$}_*\mapsto\frac{2}{R_*+1}\underbar{$n$}_*$\\[2mm]
\hline
\end{tabular}\label{tab:zf}
\end{table*}
\else
\begin{table}[!t]
\renewcommand{\arraystretch}{1.5}
\caption{Variable Choices for Specific Zames--Falb Proofs}
\label{table_example}
\centering
\begin{tabular}{cccc}
\hline
\bfseries Type & \bfseries $\delta_k$ & \bfseries $\gamma_k$ & \bfseries Notes\\
\hline\hline
$[0,\infty]$-slope restricted & $0$ & $1$ & \\
$[\alpha,\beta]$-slope restricted & $0$ & $1$ &\parbox{3.6cm}{Loop transformation \\$(y,u) \mapsto (\beta y - u,u-\alpha y)$}\\
Noisy composition & $\delta$ & $1$ & \\
Stiction (slope $1/\epsilon$) &$\delta$ & $\dfrac{1-\delta}{1+\delta}$ &\\[2mm]
Quasi-monotone/odd &$\dfrac{R_m-1}{R_m+1}, \dfrac{R_o-1}{R_o+1}$ & \parbox{2cm}{$R_m^{-1}, h_k\geq 0$\\$R_o^{-1},h_k<0$} & \parbox{4cm}{Notational change in def.\\ from \cite{heath_genzf}: $\underbar{$n$}_*\mapsto\frac{2}{R_*+1}\underbar{$n$}_*$}\\[2mm]
\hline
\end{tabular}\label{tab:zf}
\end{table}
\fi
\subsection{Proof of Theorem \ref{thm:rhozf_dd}}\label{sec:ddthmproof}

We begin with an elementary lemma; a similar one is used in \cite{damato_repeated} for the proof of (\ref{prop:constant_dd}).
\begin{lem}\label{lem:sectorrepeated}
For a repeated monotone nondecreasing nonlinearity $\Delta$ and with $u=\Delta(y)$, we have that
\begin{equation*}
u^{(i)}y^{(j)} + u^{(j)}y^{(i)} \leq u^{(i)}y^{(i)} + u^{(j)}y^{(j)},
\end{equation*}
and if $\phi$ is odd,
\begin{equation*}
|u^{(i)}y^{(j)} + u^{(j)}y^{(i)}| \leq u^{(i)}y^{(i)} + u^{(j)}y^{(j)},
\end{equation*}
for all indices $i, j$.
\end{lem}
\begin{proof}
Assume without loss of generality that $y^{(i)}\geq y^{(j)}$. By monotonicity and the fact that the nonlinearity is repeated, we must have that $u^{(i)} \geq u^{(j)}$. Thus:
\begin{equation*}
(u^{(i)}-u^{(j)})(y^{(i)}-y^{(j)})\geq 0 \:.
\end{equation*}
If $\phi$ is odd, then the second equation is proven by also observing that \mbox{$[u^{(i)},\:-u^{(j)}]^\tp = \Delta([y^{(i)},\:-y^{(j)}]^\tp)$} implies $(u^{(i)}+u^{(j)})(y^{(i)}+y^{(j)})\geq 0$.
\end{proof}

Let $E_{ij}$ denote the standard basis matrix. Now, given a diagonally dominant matrix $\Gamma$, define the following symmetric matrices:
\begin{align*}
\tilde \Gamma_{ij} = &\; \Gamma_{ij}(E_{ij}+E_{ji})+(|\Gamma_{ij}|+1)(E_{ii}+E_{jj}),\: i\neq j\\
\tilde \Gamma_{ii} = &\; E_{ii}\\
\tilde H_{ij}^l(z) = &\; \rho^{2l}z^{-l}(E_{ij}+E_{ji}), \: i\neq j\\
\tilde H_{ii}^l(z) = &\; \rho^{2l}z^{-l}E_{ii}.
\end{align*}

\begin{prop}
$\Delta$ satisfies the ``$(i,j)$ off-by-$l$ $\rho$-IQC'' defined by
\begin{align*}
\Pi =\; \bmat{0 & \tilde \Gamma_{ij} -{\tilde H_{ij}^l}\textup{*}  \\ \tilde \Gamma_{ij}-\tilde H_{ij}^l & 0}
\end{align*}
for all $\rho$ in $(0,1]$.
\end{prop}

\begin{proof}
As before, assume the $\phi$ is the gradient of a potential function $f$.
For notational convenience, define the following symbols:
\begin{align*}
d_k = &\; u_k^{(i)}y_k^{(i)} + u_k^{(j)}y_k^{(j)} \\
c_k = &\; u_k^{(j)}y_k^{(i)} + u_k^{(i)}y_k^{(j)} \\
p_k^l = &\; u_k^{(j)}y_{k-l}^{(i)} + u_k^{(i)}y_{k-l}^{(j)}\\
f_{k} = &\; f(y_k^{(i)}) + f(y_k^{(j)})\:.
\end{align*}
Using convexity and the fact that the nonlinearity is repeated, we can obtain the inequalities
\begin{align*}
d_k-p_k^l \geq &\; f_k - f_{k-l} \\
d_k \geq &\; f_k\:.
\end{align*}
Then:
\if\MODE1
\begin{align*}
&\; \sum_{k=0}^T \rho^{-2k}u_k^\tp(\tilde \Gamma_{ij}-\tilde H_{ij}^l)y_k \\
= &\; \sum_{k=0}^T \rho^{-2k}|\Gamma_{ij}|(d_k-c_k)\\
&+ \sum_{k=0}^T \rho^{-2k}((1-\rho^{2l})d_k + \rho^{2l}(d_k-p_k^{ij,l})) \\
\geq &\; \sum_{k=0}^T \rho^{-2k}|\Gamma_{ij}|(d_k-c_k)\\
&+ \sum_{k=0}^T \rho^{-2k}((1-\rho^{2l})f_k + \rho^{2l}(f_k - f_{k-l}))\:.
\end{align*}
\else
\begin{align*}
\sum_{k=0}^T \rho^{-2k}u_k^\tp(\tilde \Gamma_{ij}&-\tilde H_{ij}^l)y_k\\
&= \sum_{k=0}^T \rho^{-2k}|\Gamma_{ij}|(d_k-c_k)
+ \sum_{k=0}^T \rho^{-2k}((1-\rho^{2l})d_k + \rho^{2l}(d_k-p_k^{ij,l})) \\
&\geq \sum_{k=0}^T \rho^{-2k}|\Gamma_{ij}|(d_k-c_k)
+ \sum_{k=0}^T \rho^{-2k}((1-\rho^{2l})f_k + \rho^{2l}(f_k - f_{k-l}))\:.
\end{align*}
\fi
The first sum is nonnegative by Lemma \ref{lem:sectorrepeated} and the second sum is nonnegative by the same arguments as in Appendix \ref{sec:thm20}.
\end{proof}
We will now show that for a $(\rho,H)$-diagonally dominant matrix $\Gamma$, the $\rho$-IQC defined by
\begin{align*}
\Pi =\; \bmat{0 & \Gamma-\hat H^*  \\ \Gamma-\hat H & 0}
\end{align*}
is a positive combination of satisfied $\rho$-IQCs and is thus a satisfied $\rho$-IQC.
Toward this end:
\if\MODE1
\begin{align*}
&\;\Gamma - \hat H(z) \\
= &\; \Gamma-\sum_{i\le j,k}\rho^{-2k}H_{ij,k}\tilde \Gamma_{ij}+[\sum_{i\le j,k}\rho^{-2k} H_{ij,k}(\tilde \Gamma_{ij} - \tilde H_{ij}^k(z))]\:.
\end{align*}
\else
\begin{align*}
\Gamma - \hat H(z) = \Gamma-\sum_{j\le i,k}\rho^{-2k}H_{ij,k}\tilde \Gamma_{ij}+[\sum_{j\le i,k}\rho^{-2k} H_{ij,k}(\tilde \Gamma_{ij} - \tilde H_{ij}^k(z))]\:.
\end{align*}
\fi
The term in brackets is a nonnegative linear combination of satified IQCs, so let us focus on the first term:
\if\MODE1
\begin{align*}
&\;Q \triangleq \Gamma-\sum_{i\leq j,k}\rho^{-2k}H_{ij,k}\tilde \Gamma_{ij} \\
= &\; \begin{cases}
\Gamma_{ii} - \sum_k[\sum_{j=1, j\neq i}^n\rho^{-2k}H_{ij,k}(|\Gamma_{ij}|+1) + \rho^{-2k}H_{ii,k}] \\
 \quad (ii \text{ indices})\\
(1-\sum_k \rho^{-2k}H_{ij,k})\Gamma_{ij} \\
 \quad (ij \text{ indices}).
\end{cases}
\end{align*}
\else
\begin{align*}
Q &\triangleq \Gamma-\sum_{i\leq j,k}\rho^{-2k}H_{ij,k}\tilde \Gamma_{ij} \\
&= \begin{cases}
\Gamma_{ii} - \sum_k[\sum_{j=1, j\neq i}^n\rho^{-2k}H_{ij,k}(|\Gamma_{ij}|+1) + \rho^{-2k}H_{ii,k}] \\
 \quad (ii \text{ indices})\\
(1-\sum_k \rho^{-2k}H_{ij,k})\Gamma_{ij} \\
 \quad (ij \text{ indices}).
\end{cases}
\end{align*}
\fi
We will now show that this constant matrix $Q$ is diagonally dominant:
\if\MODE1
\begin{align*}
&Q_{ii} - \sum_{j=1,j\neq i}^n|Q_{ij}| \\
= &\; \Gamma_{ii} - \sum_k[\sum_{j=1, j\neq i}^n\rho^{-2k}H_{ij,k}(|\Gamma_{ij}|+1) + \rho^{-2k}H_{ii,k}] \\
&- \sum_{j=1,j\neq i}^n(1-\sum_k \rho^{-2k}H_{ij,k})|\Gamma_{ij}|] \\
= &\; \Gamma_{ii} - \sum_{j=1,j\neq i}^n|\Gamma_{ij}| - \sum_{j=1}^n \sum_k \rho^{-2k}H_{ij,k} \geq 0
\end{align*}
\else
\begin{align*}
Q_{ii} - \sum_{j=1,j\neq i}^n|Q_{ij}|
&= \Gamma_{ii} - \sum_k\biggl(\sum_{j=1, j\neq i}^n\rho^{-2k}H_{ij,k}(|\Gamma_{ij}|+1) + \rho^{-2k}H_{ii,k}\biggr)\\
&\qquad- \sum_{j=1,j\neq i}^n(1-\sum_k \rho^{-2k}H_{ij,k})|\Gamma_{ij}|] \\
&= \Gamma_{ii} - \sum_{j=1,j\neq i}^n|\Gamma_{ij}| - \sum_{j=1}^n \sum_k \rho^{-2k}H_{ij,k} \geq 0
\end{align*}
\fi
by the assumption that $\Gamma$ is $(\rho,H)$-diagonally dominant. Similar modifications to the assumptions and proof hold for odd and $[\alpha,\beta]$-sector nonlinearities. Also, note that in the case of diagonal (but not necessarily repeating) $\Delta$, Lemma \ref{lem:sectorrepeated} will not hold in general. Thus, considering where it is used in the proof, we would need to constrain $\Gamma$ and $H$ to be diagonal.

\subsection{Computational considerations}\label{sec:appendixcomp}
\subsubsection{Homogeneity}
We may leverage the structure of the repeated Zames--Falb $\rho$-IQCs to reduce the size and complexity of the LMI~\eqref{eq:explmi}.
\begin{prop}(Homogeneity simplification)
If we are searching over a combination of repeated Zames--Falb IQCs with fixed $\alpha,\beta,\rho$ and varying over $\Gamma$ and $H$ matrices, and fixed non-Zames--Falb $\rho$-IQCs, i.e.
\if\MODE1
\begin{align*}
    \Pi_A = \biggl\{&\; \sum_{\theta=1}^n \lambda_\theta \Pi_{ZF}(\Gamma_\theta, H_\theta) + \sum_{\delta=1}^r \lambda_\delta\Pi_\delta\: | \: \lambda_\theta \geq 0,\:\lambda_\delta \geq 0,\\
    &\; \Gamma_\theta \text{ is $(\rho, H_\theta)$-diag. dom.}\biggr\}\:.
\end{align*}
\else
\begin{align*}
    \Pi_A = \biggl\{&\; \sum_{\theta=1}^n \lambda_\theta \Pi_{ZF}(\Gamma_\theta, H_\theta) + \sum_{\delta=1}^r \lambda_\delta\Pi_\delta\: | \: \lambda_\theta \geq 0,\:\lambda_\delta \geq 0,
    \Gamma_\theta \text{ is $(\rho, H_\theta)$-diag. dom.}\biggr\}\:.
\end{align*}
\fi
Then, searching over the associated LMI is equivalent to searching over
\if\MODE1
\begin{align*}
    \Pi_B = \biggl\{ &\; \lambda \Pi_{ZF}(\Gamma, H) + \sum_{\delta=1}^r \lambda_\delta \Pi_\delta \: | \: \lambda \geq 0,\:\lambda_\delta \geq 0,\\
    \; &\; \Gamma \text{ is $(\rho, H)$-diag. dom.}\biggr\}\:.
\end{align*}
\else
\begin{align*}
    \Pi_A = \biggl\{&\; \sum_{\theta=1}^n \lambda_\theta \Pi_{ZF}(\Gamma_\theta, H_\theta) + \sum_{\delta=1}^r \lambda_\delta\Pi_\delta\: | \: \lambda_\theta \geq 0,\:\lambda_\delta \geq 0, \Gamma_\theta \text{ is $(\rho, H_\theta)$-diag. dom.}\biggr\}\:.
\end{align*}
\fi
That is, $\Pi_A=\Pi_B$.
\end{prop}

\begin{proof}
$\Pi_B \subseteq \Pi_A$ is immediate. To check the other direction, first assume $\sum_\theta \lambda_\theta = \Lambda > 0$ (the $\Lambda = 0$ case is trivial).

Now, note that, due to the linearity of the repeated Zames--Falb IQC in terms of $\Gamma$ and $H$, see that
\begin{align*}
\sum_{\theta=1}^n \lambda_\theta \Pi_{ZF}(\Gamma_\theta,H_\theta) = &\; \Lambda\: \Pi_{ZF}\left(\frac{\sum_\theta \lambda_\theta \Gamma_\theta}{\Lambda}, \frac{\sum_\theta \lambda_\theta H_\theta}{\Lambda}\right)\:.
\end{align*}
All that remains is to check that the diagonal dominance definitions are satisfied; two require care. First, the filter condition on $H$:
\begin{align*}
\sum_k \rho^{-2k}\left|\frac{\sum_\theta \lambda_\theta H_{ij,k}^\theta}{\Lambda}\right| \leq &\; \sum_k \rho^{-2k}\frac{\sum_\theta \lambda_\theta |H_{ij,k}^\theta|}{\Lambda}  \\
= &\; \frac{\sum_\theta \lambda_\theta}{\Lambda} \sum_k \rho^{-2k} |H_{ij,k}^\theta|\\
\leq &\;  \frac{\sum_\theta \lambda_\theta}{\Lambda} \quad \text{(by assumption)} \\
= &\; 1\:.
\end{align*}
Second, the main diagonal dominance condition:
\if\MODE1
\begin{align*}
& \frac{\sum_\theta \lambda_\theta \Gamma_{ii}^\theta}{\Lambda}  -  \sum_{j=1,j\neq i}^n\left|\frac{\sum_\theta \lambda_\theta \Gamma_{ij}^\theta}{\Lambda}\right| - \sum_{j=1}^n \sum_k \rho^{-2k}\left|\frac{\sum_\theta \lambda_\theta H_{ij,k}^\theta}{\Lambda}\right| \\
\geq &\; \frac{\sum_\theta \lambda_\theta \Gamma_{ii}^\theta}{\Lambda}  - \frac{\sum_\theta \lambda_\theta}{\Lambda}\sum_{j=1,j\neq i}^n| \Gamma_{ij}^\theta| - \frac{\sum_\theta \lambda_\theta}{\Lambda}\sum_{j=1}^n \sum_k \rho^{-2k}| H_{ij,k}^\theta| \\
= &\; \Lambda^{-1}\sum_\theta \lambda_\theta\left( \Gamma_{ii}^\theta  - \sum_{j=1,j\neq i}^n| \Gamma_{ij}^\theta| - \sum_{j=1}^n \sum_k \rho^{-2k}| H_{ij,k}^\theta|\right) \\
\geq &\; 0, \:\text{by assumption.}
\end{align*}
\else
\begin{align*}
\hspace{3cm}&\hspace{-3cm} \frac{\sum_\theta \lambda_\theta \Gamma_{ii}^\theta}{\Lambda}  -  \sum_{j=1,j\neq i}^n\left|\frac{\sum_\theta \lambda_\theta \Gamma_{ij}^\theta}{\Lambda}\right| - \sum_{j=1}^n \sum_k \rho^{-2k}\left|\frac{\sum_\theta \lambda_\theta H_{ij,k}^\theta}{\Lambda}\right| \\
\geq &\; \frac{\sum_\theta \lambda_\theta \Gamma_{ii}^\theta}{\Lambda}  - \frac{\sum_\theta \lambda_\theta}{\Lambda}\sum_{j=1,j\neq i}^n| \Gamma_{ij}^\theta| - \frac{\sum_\theta \lambda_\theta}{\Lambda}\sum_{j=1}^n \sum_k \rho^{-2k}| H_{ij,k}^\theta| \\
= &\; \Lambda^{-1}\sum_\theta \lambda_\theta\left( \Gamma_{ii}^\theta  - \sum_{j=1,j\neq i}^n| \Gamma_{ij}^\theta| - \sum_{j=1}^n \sum_k \rho^{-2k}| H_{ij,k}^\theta|\right) \\
\geq &\; 0, \:\text{by assumption.}
\end{align*}
\fi
The non-Zames--Falb IQCs carry straight through in both directions. Thus, $\sum_{\theta=1}^n \lambda_\theta \Pi_{ZF}(\Gamma_\theta,H_\theta) + \sum_\delta\lambda_\delta\Pi_\delta \in \Pi_A$, and therefore $\Pi_A = \Pi_B$.
\end{proof}

\subsubsection{Convexification for repeated nonlinearities}
For repeated nonlinearities, the main LMI (\ref{eq:explmi}) with constraints can be written as 
\begin{align*}
&\; \min_{P,\lambda,\Gamma,H} 0 \\
\text{s.t.} &\; \mathcal{A}(P) + \: \lambda\mathcal{M}(\Gamma, H) \prec 0\\
&\; P \succeq 0\\
&\; \lambda \geq 0\\
&\;\Gamma_{ii}\geq 0, \: \forall\:i \\
&\;\Gamma_{ij}\leq 0,  , \: \forall\: i\neq j \\
&\; H_{ij,k}\geq 0 , \: \forall\: i,j,k\\
&\;\sum_{k=0}^\infty \rho^{-2k} |H_{ij,k}| \leq 1, \: \forall\: i,j \\
&\; \Gamma_{ii} \geq \sum_{j=1, j\neq i}^n |\Gamma_{ij}| + \sum_{j=1}^n\sum_{k=0}^\infty\rho^{-2k} |H_{ij,k}|, \quad \forall\:i
\end{align*}
for some known linear functions $\mathcal{A}, \mathcal{M}$ (note that the latter is linear in the pair $[\Gamma, H]$). This problem is not immediately convex, due to the product of $\lambda$ and $[\Gamma,H]$. However, this program is indeed equivalent to the convex problem
\begin{align*}
&\; \min_{\tilde P,\zeta,\tilde\Gamma,\tilde H} 0 \\
\text{s.t.} &\; \mathcal{A}(\tilde P) +  \mathcal{M}(\tilde\Gamma, \tilde H) \preceq -I\\
&\; \tilde P \succeq 0\\
&\; \zeta > 0\\
&\; \tilde\Gamma_{ii}\geq 0, \: \forall\:i \\
&\; \tilde\Gamma_{ij}\leq 0,  , \: \forall\: i\neq j \\
&\; \tilde H_{ij,k}\geq 0 , \: \forall\: i,j,k\\
&\;\sum_{k=0}^\infty \rho^{-2k}|\tilde H_{ij,k}| \leq \zeta, \: \forall\: i,j \\
&\; \tilde\Gamma_{ii} \geq \sum_{j=1, j\neq i}^n |\tilde\Gamma_{ij}| + \sum_{j=1}^n\sum_{k=0}^\infty\rho^{-2k}|\tilde H_{ij,k}|, \quad \forall\:i\:.
\end{align*}
In practice, it is helpful to replace $\preceq -I$ with $\preceq -tI$ and maximize over $t$, while placing upper bounds on $\zeta$ and $t$.
%%%%%%%%%%%%%%%%%%%%%%%%%%%%%%%%%%%%%%%%%%%%%%%%%%%%%%%%%%%%%%%%%%%%%%%%%%%%%%%%

\if\MODE1
\bibliographystyle{IEEEtran}
\bibliography{freq}
\else
\begin{small}
\bibliographystyle{abbrv}
\bibliography{freq}
\end{small}
\fi

\end{document}